\documentclass[12pt,3p]{elsarticle}
\usepackage{amsmath}
\usepackage{amsfonts}
\usepackage{amssymb}
\usepackage{amsthm}
\usepackage{empheq}

\usepackage{tikz}
\usepackage{pgf}
\usetikzlibrary{arrows}

\theoremstyle{plain}
\newtheorem{theorem}{Theorem}[section]
\newtheorem{proposition}[theorem]{Proposition}
\newtheorem{lemma}[theorem]{Lemma}

\theoremstyle{definition}
\newtheorem{definition}[theorem]{Definition}

\theoremstyle{remark}
\newtheorem{remark}[theorem]{Remark}

\numberwithin{equation}{section} %% Equation numbering control.
\newcommand{\vect}[1]{\mathbf{#1}}

\newcommand{\bk}{\vect{k}}

\newcommand{\bu}{\vect{u}}
\newcommand{\bv}{\vect{v}}
\newcommand{\bw}{\vect{w}}

\newcommand{\bx}{\vect{x}}
\newcommand{\by}{\vect{y}}

\newcommand{\bg}{\vect{g}}

\newcommand{\bbf}{\vect{f}}

\newcommand{\field}[1]{\mathbb{#1}}
\newcommand{\nN}{\field{N}}
\newcommand{\nZ}{\field{Z}}
\newcommand{\nQ}{\field{Q}}
\newcommand{\nR}{\field{R}}

\newcommand{\cA}{\mathcal A}

\newcommand{\cD}{\mathcal D}

\newcommand{\nT}{\mathbb T}
\newcommand{\vphi}{\varphi}
\newcommand{\maps}{\rightarrow}

\newcommand{\sand}{\quad\text{and}\quad}

\newcommand{\pd}[2]{\frac{\partial #1}{\partial #2}}
\newcommand{\abs}[1]{\left\lvert#1\right\rvert}
\newcommand{\set}[1]{\left\{#1\right\}}
\newcommand{\ip}[2]{\left<#1,#2\right>}
\newcommand{\iip}[2]{\left<\left<#1,#2\right>\right>}

\newcommand{\pnt}[1]{\left(#1\right)}

\newcommand{\diff}[1]{\widetilde{#1}}
\newcommand{\bud}{\diff{\bu}}

\newcommand{\thetad}{\diff{\theta}}

\newcommand{\initialData}[1]{#1_0}
\newcommand{\buInit}{\initialData{\bu}}

\newcommand{\thetaInit}{\initialData{\theta}}

\newcommand{\comments}[1]{}
\newcommand{\nn}{\nonumber}
\newcommand{\ra}{\rightarrow}
\newcommand{\lra}{\longrightarrow}

\newcommand{\h}{\mathbb H}

\title{On the Attractor for the Semi-Dissipative Boussinesq Equations}
\date{\today}

\author[umbc]{Animikh Biswas\corref{cor1}}
\address[umbc]{Department of Mathematics \& Statistics,
        University of Maryland, Baltimore County,
        1000 Hilltop Circle.
        Baltimore, MD 21250, USA}
\ead{abiswas@umbc.edu}
\cortext[cor1]{Corresponding author}
\author[tamu]{Ciprian Foias}
\address[tamu]{Department of Mathematics,
                Texas A\&M University,
        College Station, TX 77843, USA}
\ead{foias@math.tamu.edu}
\author[unl]{Adam Larios}
\address[unl]{Department of Mathematics, 
                University of Nebraska--Lincoln,
        Lincoln, NE 68588-0130, USA}
\ead{alarios@unl.edu}

\begin{document}


%==============================================================-
\begin{abstract}
In this article, we study the long time behavior of solutions of a variant of the Boussinesq system in which the equation for the velocity is parabolic while the equation for the temperature is hyperbolic. We prove that the system has a global attractor which retains some of the properties of the global attractors for the 2D and 3D Navier-Stokes equations. Moreover, this attractor contains infinitely many invariant manifolds in which several universal properties of the Batchelor, Kraichnan, Leith theory of turbulence are potentially present.

\comments{We propose a notion of a global attractor for the 2D semi-dissipative Boussinesq equations, which is a mixed parabolic-hyperbolic system.  This new notion of attractor generalizes the usual definition given in the case of the 2D Navier-Stokes equations.  We prove that this attractor retains some of the properties of attractor for the 2D Navier-Stokes equations, albeit in a weaker sense.  We also prove that smooth solutions to the 2D semi-dissipative Boussinesq equations have the property of backward uniqueness.}

\noindent
 {\sc R\'esum\'e.} Dans cet article nous \'etudions le comportment au temps infini des solutions d'une partiellement dissipative system du  Boussinesq, dont une est parabolique et l'autre est hyperbolique. Dans ce but nous introduire un attracteur universel qui retient plusieurs propriet\'es des attracteurs universels des
 \'equations de Navier-Stokes en dimension deux ou trois, qui, in particulier, contient une infinit\'e de 
 variet\'es invariantes dans lesquelles plusieurs propriet\'es universelles de la th\'eorie de la turbulence bidimensionnelle de Batchelor, Kraichnan et Leith, sont potentiellement pr\'esentes.
\end{abstract}

\maketitle
\thispagestyle{empty}%Gets rid of page number on first page.
%============================================================

%\newpage
\noindent
\tableofcontents
% =====================================================================
\section{Introduction}\label{sec:Int}
% =====================================================================

\noindent
A critical part of understanding turbulence in geophysical flows lies in understanding the large-time dynamics of the system.  A tool that provides deep insight in studying large-time dynamics of certain systems is the study of global attractors, which typically requires the underlying dynamical system to be dissipative; see for instance \cite{CFT85,  Constantin_Foias_1988, Robinson_2001, Temam_1997_IDDS} and references therein.  For instance, in case of the fully dissipative Boussinesq system, the study of the long term dynamics via global attractors has been accomplished in \cite{FMT}.
 However, in certain physical regimes, the dynamical system governing geophysical flows is modeled to be only semi-dissipative.  That is, the system is dissipative in some variables, but not in others.  This means that the usual notion of global attractors no longer applies.  In this work, we propose a new type of global attractor---which we call a weak sigma-attractor---that captures the  large-time dynamics of a well-known semi-dissipative model for ocean flows known as the semi-dissipative 2D Boussinesq system
 \cite{Larios_Lunasin_Titi_2013}. The notion of the {\it weak sigma-attractor}, which in some respects 
 is a hybrid of the attractors of the 2D and 3D Navier-Stokes equations
 \cite{FRT2010,FT1987}, exploits the fact that  while only part of the system is dissipative, the other part has a hyperbolic structure that gives rise to certain conserved quantities which in turn correspond to invariant sets in the phase space.  
 We expect that this approach, which is new to the best of our knowledge,  will be useful in studying the long-term dynamics of many other semi-dissipative systems.  

% The global attractor for a dynamical system is usually defined as a set $\mathcal{A}$ in the phase space having the following properties.
% \begin{enumerate}
%  \item $\mathcal{A}$ is the smallest set which
% uniformly attracts all bounded sets.
% \item $\mathcal{A}$ is the largest compact invariant set.
% \item $\mathcal{A}$ is the set of all points on (extended) trajectories which are bounded backward
% and forward in time.
% \end{enumerate}

The global attractor for a dissipative dynamical system is usually defined as either the maximal compact invariant set, the minimal set which uniformly attracts all bounded sets, or the set of points on complete bounded trajectories.  For many dissipative systems (e.g., the heat equation, the 2D Navier-Stokes equations, the 1D Kuramoto-Sivashinsky equations, and many others), these definitions are equivalent.  For the semi-dissipative Boussinesq system, this equivalence fails, and therefore the definition must be modified.  In particular, compactness must be sacrificed.  However, as we will show, it can be recovered in a certain local sense, as the attractor we propose will be shown to be a countable union of weakly compact invariant sets that retain several properties of the attractor for the 2D Navier-Stokes equations.  Moreover, the attraction properties we describe will be in the sense of the weak topology rather than the strong topology.  This aspect is similar to the situation for the 3D Navier-Stokes equations \cite{FRT2010,FT1987}, where, due to the lack of knowledge about the global well-posedness of the system, one studies a weak attractor rather than a strong global attractor.  However, the system we study is known to be globally well-posed, 
and thus, the use of the weak topology is intrinsic to the semi-dissipative nature of the system under consideration. In contrast to this, if the global well-posedness for the 3D Navier-Stokes equations  were to be established, the corresponding weak attractor will immediately become an  attractor in the usual (strong) sense.

A curious feature of the 2D semi-dissipative Boussinesq system \eqref{Bouss}  pertaining to the study of 2D turbulence deserves mention here.   In particular, the equations potentially combine a wealth of turbulent dynamics in a single system.  This is due to the fact that the right-hand side of the momentum equation is a time-dependent dynamical variable, but that its $L^p$ norms are time-independent ($1\leq p\leq\infty$).  Thus, the velocity component of  \eqref{Bouss} will indeed contain turbulent dynamics corresponding to all Grashof numbers, so long as one can establish that arbitrarily large  \textit{effective Grashof number}--defined as the magnitude of the divergence-free part of the force that drives the velocity equation--can be achieved for large time.  We discuss this issue in greater detail in Section \ref{sec:2dturb}.
On the invariant  subsets mentioned above, the system is expected to have time statistics consistent with the Batchelor-Kraichnan-Leith \cite{Bat, Kr} empirical theory of 2D turbulence, albeit with different Grashof numbers varying from subset to subset. We can therefore think of the 2D semi-dissipative Boussinesq system as a platform which hosts an enormous variety of 2D turbulent dynamics.  It appears to be a model system for the study of attractors, in the same way that C. Elegans and Drosophilia are model species for study in biology.  (Similar statements may be made in the 3D case.)  We emphasize that this system  was not created in an \textit{ad-hoc} fashion to have these particular properties, but instead arises naturally as a relevant system in geophysics.  We will show that for system \eqref{Bouss}, the weak sigma-attractor is not only non-trivial, but it is an extremely rich proper subset of the phase space. We hope that the study of the long-term dynamics of this system will inspire further research leading to greater insights into 2D turbulence, and perhaps suggest new approaches to this area.

The paper is organized as follows. In Section \ref{system}, we describe the semi-dissipative Boussinesq system, which is the main subject of our study.
In Section \ref{sec_prelim}, we lay out notation and preliminary material.
In Section \ref{ssec_well_posedness}, we recall well-known well-posedness results which will be useful in the sequel.
In Section \ref{sec_Semigroup}, we prove that smooth solutions to system \eqref{Bouss} have the backward-uniqueness property.  \comments{The proof uses the invention of a non-standard quantity combining the variables of the hyperbolic and parabolic parts of the system. }
In Section \ref{sec_notion_of_attractor}, we propose the notion of a weak sigma-attractor for system \eqref{Bouss}, describe several new classes of steady states (as well as time variant solutions) of the system, and state the main theorem on the properties of the weak sigma-attractor.  
In Section \ref{sec_main_proof}, we prove the properties discussed in Section \ref{sec_notion_of_attractor}.
In Section \ref{sec:2dturb}, we discuss connections to turbulence as evidenced in the 2D Navier-Stokes equations with forcing in {\it all} Fourier modes, while in Section \ref{sec_Energy_Enstrophy}, we discuss the projection of  the attractor onto the energy-enstrophy plane.  We find certain points in this setting which correspond to steady columnar flows which are in hydrostatic balance.
In Section \ref{sec_open_questions}, we discuss some open questions related to this study.  
Section \ref{sec_appendix} is an appendix, where some standard facts are collected for reference of the reader.

\subsection{Semi-dissipative Boussinesq system}  \label{system}

The semi-dissipative 2D Boussinesq system (without rotation) in the
periodic  domain $\Omega:=[0,L]^2$ for time $t \ge 0$ is given by
\begin{subequations}\label{Bouss}
\begin{alignat}{2}
\label{Bouss_mo}
\partial_t\bu + (\bu\cdot\nabla)\bu +\nabla p&= \nu\triangle\bu+ \theta \bg,
\quad&& \text{in }\Omega\times \nR_+,\\
\label{Bouss_den}
\partial_t\theta + (\bu\cdot\nabla)\theta \;\;\,\quad\quad &=0,
\quad&& \text{in }\Omega\times \nR_+,\\
\label{Bouss_div}
\nabla \cdot \bu &=0,
\quad&& \text{in }\Omega\times\nR_+,\\
\label{Bouss_init}
\bu(\bx,0)=\buInit(\bx),
\quad\theta(\bx,0)&=\thetaInit(\bx),
\quad&& \text{in }\Omega,
\end{alignat}
\end{subequations}
equipped with periodic boundary conditions in space \cite{Larios_Lunasin_Titi_2013}.  
Here $\nu>0$ is the fluid viscosity.  
The spatial variable is denoted
$\bx=(x_1,x_2)\in\Omega$, and the unknowns are the fluid velocity field
$\bu\equiv\bu(\bx,t)\equiv(u_1(\bx,t),u_2(\bx,t))$, the fluid pressure
$p(\bx,t)$, and the function $\theta\equiv\theta(\bx,t)$, which may be
interpreted physically, e.g., as the temperature variable.  
We write
$\bg=(0,g)^T$ for the constant, upward-pointing gravity vector, where $g$
is the (scalar) acceleration due to gravity.  It is straight-forward to show the Galilean invariance and mean-preserving property of \eqref{Bouss} (see Section \ref{ssec_Galilean_Invariance}).  Therefore, we henceforth assume that $\int_\Omega\bu\,d\bx = 0$ and $\int_\Omega\theta\,dx=0$, and similarly for the initial data.  Furthermore, all function spaces used in this work will be assumed to contain only mean-free elements, unless otherwise indicated.
We remark that, with minimal adjustments, the results here are also valid in the presence of a
Coriolis rotational term.

% Our new notion of attractor, which applies to \eqref{Bouss}, is somewhat analogous to the attractor for the 2D Navier-Stokes equations, although it differs in significant ways.    
% For example, note that any suitable definition of such a set must include all the steady states of the system.  
% Recently, the authors found two very rich infinite-dimensional orthogonal
% families of steady state solutions.  
% This means that our analogue of the conventional attractor must allow for
% infinite-dimensional subsets.  
% However, \eqref{Bouss_den} imposes strong restrictions on the solutions, which we exploit to aid in our understanding of conservative sets.  
% In particular, the many conservation laws inherent in \eqref{Bouss_den} separate the phase-space into (not-necessarily disjoint) subsets, each of which has its own characteristic dynamics.
% It is a natural conjecture that, on each of these subsets, the invariant manifolds are those on
% which the topological entropy is finite.
% 
% We will show that system \eqref{Bouss} has an ``absorbing cylinder'' $\Gamma$ with an infinite-dimensional compact base. 
% This $\Gamma$ will contain the supports of all the invariant
% probability measures (the existence of which has yet to be rigorously proven)
% that give rise to the same universal time statistics.  
% Furthermore, our putative global attractor will be a subset of $\Gamma$.

% =====================================================================
 \section{Notation and Some Specific Preliminaries}\label{sec_prelim}
% =====================================================================
\noindent 
 \comments{ In fairness to the reader, we must discuss the abbreviated functional notation we will use throughout the proofs. Most of those proofs are based on calculations with long formulas and estimates. Using the classical notation would substantially increase the space necessary for those proofs. Therefore we introduce the functional notation in the preliminary section and, in order to help the reader become familiar with the notation, we use it simultaneously with the classical notation in that section.}
%In this section, we lay out notation and preliminary material.
Let $\mathcal{F}$ be the set of all vector-valued trigonometric polynomials with periodic domain $\Omega=\nT^2:=\nR^2/(L\nZ^2)\equiv[0,L]^2$.  We define a space of test functions
\[\mathcal{V}:=\set{\vphi\in\mathcal{F}:\nabla\cdot\vphi=0\text{ and  }\int_{\nT^2}\vphi(x)\,dx=0}.\]
We denote by $L^2$ and $\h^m$ the usual Lebesgue and Sobolev spaces over $\nT^2$, and define $H$ and $V$ to be the closures of $\mathcal{V}$ in $L^2$ and $\h^1$ respectively.  We will often use the notation $(\cdot,\cdot)$ to denote pairs, so to avoid confusion, we denote the inner products on $H$ and $V$ respectively by
\[\ip{\bu}{\bv}:=\sum_{i=1}^2\int_{\nT^2} u_iv_i\,dx
\sand
\iip{\bu}{\bv}:=\sum_{i,j=1}^2\int_{\nT^2}\pd{u_i}{x_j}\pd{v_i}{x_j}\,dx,
\]
and the associated norms, respectively by
\[
\|\bu\|_{L^2}:=\ip{\bu}{\bu}^{1/2},\quad
\|\bu\|_{\h^1}:=\iip{\bu}{\bu}^{1/2}. 
\]
 This notation will also be extended to apply to tensors, such as $\nabla\bu$, in the natural way, which should not be a source of confusion.  Note that $\|\cdot\|$ is a norm due to the Poincar\'e inequality, \eqref{poincare}, below.
We denote by $V'$ the dual space of $V$.  The action of $V'$ on $V$ is denoted by $\ip{\cdot}{\cdot}_{V',V}$.  Note that we have the continuous embeddings $V\hookrightarrow H\hookrightarrow V'$.
Moreover, by the Rellich-Kondrachov compactness theorem (see, e.g., \cite{Adams_Fournier_2003,Evans_2010}), these embeddings are compact due to the boundedness of the domain $\nT^2$.

We denote by $P_\sigma:L^2\maps H$ the Leray-Helmholtz projection operator (i.e., the orthogonal projection onto divergence-free vector spaces), and the Stokes operator $A:=-P_\sigma\triangle$ with domain $\mathcal{D}(A):=\h^2\cap V$.   In our case of periodic boundary conditions, it is well-known that $A=-\triangle$  (see, e.g., \cite{Constantin_Foias_1988,Temam_1995_Fun_Anal}).  $A^{-1}:H\maps H$ is a positive-definite, self-adjoint, compact operator, and therefore has
an orthonormal basis of eigenfunctions $\vphi_k$ corresponding to a
non-increasing sequence of eigenvalues (see, e.g.,
\cite{Constantin_Foias_1988,Temam_1995_Fun_Anal}).  We label the
eigenvalues $\lambda_k$ of $A$ so that
$0<\lambda_1\leq\lambda_2\leq\lambda_3\leq\cdots$.  Furthermore, for all $\bw\in V$, we have
the Poincar\'e inequality,
\begin{align}\label{poincare}
 \kappa_0\|\bw\|_{L^2}\leq \|\bw\|_{\h^1},\qquad \kappa_0:= 2\pi/L.
\end{align}
Due to \eqref{poincare}, for $\bw\in \mathcal{D}(A)$, we have the norm equivalences
\begin{equation}\label{elliptic_reg}
   \|A \bw\|_{L^2}\cong\| \bw\|_{\h^2}
   \sand
  \|\nabla \bw\|_{L^2}\cong\|\bw\|_{\h^1}.
\end{equation}
In fact, for any $s\in\nR$, one can show that
\begin{equation}
   \|\bw\|_{\h^s} \cong \|A^{s/2} \bw\|_{L^2}
\end{equation}
for $\bw\in\mathcal{D}(A^{s/2})$, where $\mathcal{D}(A^{s/2})$ and $A^{s/2}$ are defined in a natural way by the eigenvalues and eigenfunctions of $A$ (see, e.g.,
\cite{Constantin_Foias_1988,Temam_1995_Fun_Anal}).

We also note the Ladyzhenskaya inequality, which in 2D reads
\begin{align}\label{ladyzhenskaya}
 \|B(\bu,\bu)\|_{L^2} \le C\|\bu\|_{L^2}^{1/2}\|A\bu\|_{L^2}^{1/2}\|\bu\|_{\h^1},
\end{align}
and the Agmon inequlaity, which in 2D reads
\begin{align}\label{agmon}
\|\bu\|_{L^{\infty}} \le C\|\bu\|_{L^2}^{1/2}\|A\bu\|_{L^2}^{1/2}.
\end{align}

We use the standard notation
\begin{equation}\label{Bdef}
 B(\bw_1,\bw_2):=P_\sigma((\bw_1\cdot\nabla)\bw_2)
\end{equation}
for $\bw_1,\bw_2\in\mathcal{V}$.  
We note several important properties of $B$ which are proven, e.g., in \cite{Constantin_Foias_1988, Foias_Manley_Rosa_Temam_2001, Temam_1995_Fun_Anal, Temam_2001_Th_Num}.
The operator $B$ defined in \eqref{Bdef} is a bilinear form which can be extended as a continuous map $B:V\times V\maps V'$.  Furthermore, for $\bw_1$, $\bw_2$, $\bw_3\in V$,
\begin{subequations}\label{symmetry}
\begin{align}
 \ip{B(\bw_1,\bw_2)}{\bw_3}_{V',V}&=-\ip{B(\bw_1,\bw_3)}{\bw_2}_{V',V},
\\\label{Bzero}
 \ip{B(\bw_1,\bw_2)}{\bw_2}_{V',V}&=0.
\end{align}
\end{subequations}
The following inequalities hold for smooth functions:
\begin{subequations}
\begin{align}
 \label{B442}
 |\ip{B(\bw_1,\bw_2)}{\bw_3}|
 &\leq
 C\|\bw_1\|_{L^2}^{1/2}\|\bw_1\|_{\h^1}^{1/2}
\|\bw_2\|_{L^2}^{1/2}\|\bw_2\|_{\h^1}^{1/2}
 \|\bw_3\|_{\h^1}
 \\
  \label{B424}
 |\ip{B(\bw_1,\bw_2)}{\bw_3}|
 &\leq
 C\|\bw_1\|_{L^2}^{1/2}\|\bw_1\|_{\h^1}^{1/2}
 \|\bw_2\|_{\h^1}
 \|\bw_3\|_{L^2}^{1/2}\|\bw_3\|_{\h^1}^{1/2}
  \\
  \label{Bi22}
 |\ip{B(\bw_1,\bw_2)}{\bw_3}|
 &\leq
 C\|\bw_1\|_{L^2}^{1/2}\|A\bw_1\|_{L^2}^{1/2}
 \|\bw_2\|_{\h^1}
\|\bw_3\|_{L^2}
   \\
  \label{B2i2}
 |\ip{B(\bw_1,\bw_2)}{\bw_3}|
 &\leq
 C\|\bw_1\|_{L^2}
 \|\bw_2\|_{\h^1}^{1/2}\|A^{3/2}\bw_2\|_{L^2}^{1/2}
 \|\bw_3\|_{L^2}.
%   \\&\leq
%  C\|\bw_1\|_{L^2}
% \|A^{3/2}\bw_2\|_{L^2}
%  \|\bw_3\|_{L^2}
\end{align}
\end{subequations}
\noindent
These results also hold if $B$ is extended to allow $\bw_2$ and $\bw_3$ to be scalar-valued functions in the indicated spaces (we still require $\nabla\cdot\bw_1=0$), and $A$ is replaced by $-\triangle$.

% ==============================================================
\subsection{Known Well-Posedness Results}\label{ssec_well_posedness}
% ==============================================================

The global well-posedness of system \eqref{Bouss} in the case $(\buInit,\thetaInit)\in (\h^3\cap V)\times \h^3$ was proven in \cite{Chae_2005}, and in the case $(\buInit,\thetaInit)\in (\h^2\cap V)\times \h^3$ in \cite{Hou_Li_2005}.  The requirements on the initial data were weakened in  \cite{Danchin_Paicu_2008_French}, where the following theorem was proven (see also \cite{Hmidi_Keraani_2007,Larios_Lunasin_Titi_2013} for related results and additional discussion).
\begin{definition}  \label{def:Bouss_wk}
We say that a  pair $(\bu,\theta)$ solves \eqref{Bouss} if  $(\bu(t),\theta(t)) \, \in \, H \times L^2$ for each $t >0$,  $\bu$ and $\theta$ are both periodic and mean-free, and they satisfy
\begin{subequations}\label{Bouss_wk}
\begin{align}
&\quad\label{Bouss_wk_mo}
-\int_0^T\ip{\bu(s)}{\Phi'(s)}\,ds
+\int_0^T\ip{\bu(s)\otimes\bu(s)}{\nabla\Phi} \,ds
 \\
 &\notag
 =
\ip{\buInit}{\Phi(\cdot,0)}
-\nu\int_0^T\iip{\bu(s)}{\Phi(s)}\,ds
+\int_0^T\ip{\theta(s)\bg}{\Phi(s)}\,ds,
\\ 
&\quad
-\int_0^T\ip{\theta(s)}{\phi'(s)}\,ds
+\label{Bouss_wk_den}
\int_0^T\ip{\bu(s)\theta(s)}{\nabla\phi(s)}\,ds
=
\ip{\thetaInit}{\phi(\cdot,0)}
\end{align}
\end{subequations}
for all mean-free, space-periodic scalar test functions $\phi(x,t) \in C^\infty(\Omega \times [0,T])$, such that $\phi(x,T) =0$; and for all mean-free, space-periodic vector-valued test functions $\Phi(x,t)\in  [ C^\infty(\Omega \times [0,T])]^2$ such that $\nabla\cdot\Phi(\cdot,t) =0$, $\Phi(\cdot,T) =0$.

A pair $(\bu(t),\theta(t)) \in H\times L^2, t \in \nR$ is a global solution if 
\begin{subequations}\label{Bouss_wk_global}
\begin{align}
&\quad\label{Bouss_wk_mo_global}
-\int_{-\infty}^\infty\ip{\bu(s)}{\Phi'(s)}\,ds
+\int_{-\infty}^\infty \ip{\bu(s)\otimes\bu(s)}{\nabla\Phi} \,ds
 \\&\notag
=
-\nu\int_{-\infty}^\infty \iip{\bu(s)}{\Phi(s)}\,ds
\int_{-\infty}^{\infty}\ip{\theta(s)\bg}{\Phi(s)}\,ds,
\\ 
&\quad
-\int_{-\infty}^\infty \ip{\theta(s)}{\phi'(s)}\,ds
+\label{Bouss_wk_den_global}
\int_{-\infty}^\infty \ip{\bu(s)\theta(s)}{\nabla\phi(s)}\,ds
= 0,
\end{align}
\end{subequations}
where $\phi \in C^\infty(\Omega \times (-\infty,\infty))$ and $\Phi \in [C^\infty(\Omega \times (-\infty,\infty))]^2$
are exactly as before, except that they are assumed to be compactly supported in time.
As mentioned before, the pressure $p$ is recovered via \eqref{pressuredef}.

\end{definition}
\begin{theorem}[\cite{Danchin_Paicu_2008_French}]\label{thm_Danchin_Paicu}
 Suppose $(\buInit,\thetaInit)\in H\times L^2$.  Then, for any $T>0$, the system\footnote{in the appropriate weak formulation, \eqref{Bouss_wk} below} \eqref{Bouss} has a unique solution $(\bu,\theta)$ such that
 \begin{align*}
 &  \bu\in C([0,T],H)\cap L^2([0,T],V),
  &&  \theta\in C([0,T],L^2),
  \\
 & \pd{\bu}{t}\in L^2([0,T],V'), 
 && \pd{\theta}{t}\in L^4([0,T],\h^{-3/2}).
 \end{align*}
If additionally, $\buInit\in V$, then 
\begin{align*}
 &  \bu\in C([0,T],V)\cap L^2([0,T],\cD(A)),
  && 
  \\
 & \pd{\bu}{t}\in L^2([0,T],H), 
 && \pd{\theta}{t}\in L^4([0,T],\h^{-1}).
 \end{align*}
% $\bu\in L^2([0,T],\cD(A))\cap C([0,T],V)$, $\pd{\bu}{t}\in L^2([0,T],H)$, and
%    $\pd{\theta}{t}\in L^2([0,T],\h^{-1})$.
\end{theorem}

Recently, results on well-posedness and persistence of regularity were improved including the treatment of the bounded domain setting \cite{Hu_Kukavica_Ziane_2014,Hu_Kukavica_Ziane_2013}.  The following two theorems collects some of the results of the aforementioned works.

\begin{theorem}[Combined results from \cite{Chae_2005, Danchin_Paicu_2008_French,Hou_Li_2005,Hu_Kukavica_Ziane_2014,Hu_Kukavica_Ziane_2013}]
\label{thm_Well_Posedness1}  Let $s\geq0$. Suppose $\buInit\in \h^{1+s}\cap V$, and $\thetaInit\in \h^s$.  Then $\bu\in C([0,\infty),\h^{1+s})\cap L_{\text{loc}}^2([0,\infty),\h^{2+s})$ and $\theta\in C([0,\infty),\h^s)$.
\end{theorem}
\noindent
\begin{theorem}[Combined results from \cite{Chae_2005, Danchin_Paicu_2008_French,Hou_Li_2005,Hu_Kukavica_Ziane_2014,Hu_Kukavica_Ziane_2013}]
\label{thm_Well_Posedness2}  Let $s>1$. Suppose $\buInit\in \h^{s}\cap V$, and $\thetaInit\in \h^s$.  Then $\bu\in C([0,\infty),\h^{s})\cap L_{\text{loc}}^2([0,\infty),\h^{2+s})$ and $\theta\in C([0,\infty),\h^s)$.
\end{theorem}

Next, we recall the well-known result (see, e.g., \cite{DiPerna_Lions_1989}) from which it follows that the transport equation has unique solution which preserves all $L^p$ norms as well as  distribution functions.   Recall that for a function  $\vphi\in L^2(\Omega)$,  the corresponding distribution function
$F_{\vphi}:\nR\maps[0,1]$ defined by
\begin{align}\label{dist_def}
 F_{\vphi}(\rho):=\frac{1}{|\Omega|}\int_{\Omega}\chi_{(-\infty,\rho]}(\vphi(x))\,dx 
 = 
 \frac{1}{|\Omega|}\abs{\set{x\in\Omega:\vphi(x)\leq\rho}}.
\end{align}
\begin{proposition}\label{prop_theta_dist_conserved}
Let $\bv\in L^2(0,T;V)$ and $\thetaInit\in L^2(\Omega)$, the Cauchy
problem
\begin{subequations}\label{transport}
 \begin{empheq}[left=\empheqlbrace]{align}\label{transport_eqn}
\partial_t\theta + \nabla\cdot(\bv\theta) &=0,\\
 \theta(\bx,0) &= \thetaInit(\bx),
\end{empheq}
\end{subequations}
has a unique solution $\theta=\theta(\bx,t)$ such that for a.e. $t\geq0$ and all
$\rho\in\nR$,
 \begin{align}\label{dist_funcs_constant}
  F_{\theta(\cdot,t)}(\rho) = F_{\thetaInit}(\rho).
 \end{align}
 In particular, for all $p\in[1,\infty]$, and a.e. $t\geq0$, 
 \begin{align}\label{Lp_conserved}
  \|\theta(t)\|_{L^p} = \|\thetaInit\|_{L^p}.
 \end{align}
 \end{proposition}
 The proof is well-known, but we sketch it here for completeness.
 \begin{proof}
 The existence and uniqueness is well-known; see, e.g.,
\cite{DiPerna_Lions_1989} and the references therein.  
Equation 
\eqref{dist_funcs_constant} is also well-known, and follows by noting that, 
since $\bv\in L^2(0,T;V)$ and $\theta$ satisfies \eqref{transport_eqn}, $\theta$ is
constant along streamlines.  A change of variables (sometimes called the ``back to labels'' map), using the fact the $\nabla\cdot\bv=0$, 
then yields \eqref{dist_funcs_constant} (see, e.g., \cite{Chorin_Marsden_1993,DiPerna_Lions_1989}
for more details).  Furthermore, \eqref{Lp_conserved} holds for $p\in[1,\infty)$ by applying \eqref{dist_funcs_constant} to the well-known identity (see, e.g., \cite{Knapp_Basic_Real_Analysis_2005}, VI.10)
\begin{align}\label{distribution_identity}
 \|\vphi\|_{L^p}^p = \int_{-\infty}^\infty\rho^p\,dF_\vphi(\rho).
\end{align}
The case $p=\infty$ is also classical, and can be proven, e.g.,  by using the Hopf-Stampaccia technique (see, e.g., \cite{Kinderlehrer_Stampacchia_1980}).  Namely, one notices that $\Theta:=\theta-\|\thetaInit\|_{L^\infty}$ also satisfies \eqref{transport_eqn}, and then takes the inner product of the resulting equation with $\Theta^+$ and $\Theta^-$, the positive and negative parts of $\Theta$.  Integrating by parts and using the divergence-free condition yields $\|\Theta\|_L^2=0$.
 \end{proof}

% ===============================================
% ======= End Well-Posedness Subsection =========
% ===============================================

\subsection{A Functional Form of the equation}
As is customary for the Navier-Stokes equations, 
applying $P_\sigma$ to \eqref{Bouss_mo} (see Section \ref{sec_prelim} for notation) yields the functional equation
\begin{align}   \label{functmom}
 \frac{d\bu}{dt} + \nu A\bu + B(\bu,\bu) = P_\sigma(\theta\bg).
\end{align}
It is straight-forward to check that 
\begin{align}
 P_\sigma(\theta\bg) 
 = \theta\bg-g\nabla(\partial_{x_2}\triangle^{-1}\theta)
 =
\binom{-\partial_{x_1}\partial_{x_2}\triangle^{-1}}{I-\partial_{x_2}^2\triangle^
{-1}}\theta g.
\end{align}
As this is an orthogonal decomposition,
\begin{align}\|\theta\bg\|_{L^2}^2=\|P_\sigma(\theta\bg)\|_{L^2}^2+g^2\|\nabla(\partial_{x_2}\triangle^{-1}\theta)\|_{L^2}^2.\end{align}  
Thus,
\begin{align}\label{Riesz}
 g^{-2}\|P_\sigma(\bg\theta)\|_{L^2}^2
 &=\notag
 \|\theta\|_{L^2}^2-\pnt{\nabla\triangle^{-1}\partial_{x_2}\theta,\nabla\triangle^{-1}
\partial_{x_2}\theta}
 \\&=\notag
 \|\theta\|_{L^2}^2+\pnt{\triangle\triangle^{-1}\partial_{x_2}\theta,
 \triangle^{-1}\partial_{x_2}\theta}
%  \\&=\notag
%  \|\theta\|_{L^2}^2-\pnt{\partial_{x_2}\theta,
%  (-\triangle)^{-1}\partial_{x_2}\theta}
 \\&=\notag
 \|\theta\|_{L^2}^2-\|(-\triangle)^{-1/2}\partial_{x_2}\theta\|_{L^2}^2
  \\&=
 \|\theta\|_{L^2}^2-\|R_2\theta\|_{L^2}^2=\|R_1\theta\|_{L^2}^2,
\end{align}
where $R_i:=(-\triangle)^{-1/2}\partial_{x_i}, i=1,2$ are the  Riesz-transforms, which can also be defined by their  Fourier symbols $\hat{R}_i(\bk):=ik_i/|\bk|, i=1,2$ at wave number $\bk\equiv(k_1,k_2)$.

 Due to the mean-free assumption, if one obtains a pair $(\bu,\theta)$  which solves \eqref{functmom} and \eqref{Bouss_den}--\eqref{Bouss_init}, then one can obtain the pressure $p$ by taking divergence in \eqref{Bouss_mo} and noting that 
\begin{gather}  \label{pressuredef}
p= -  \sum_{i,j} R_{ij}(u_iu_j)+ g\triangle^{-1}(\partial_{x_2} \theta), 
\quad
R_{ij}=\triangle^{-1}(\partial_{x_i}\partial_{x_j}).
\end{gather}
where $\triangle^{-1}$ is taken with respect to the periodic boundary conditions and the mean-free condition.
Thus, we will also refer to a pair $(\bu(t),\theta(t))$ solving \eqref{functmom} and \eqref{Bouss_den}--\eqref{Bouss_init} as a solution to the system \eqref{Bouss}.

% It will be convenient to define the (scalar) vorticity by
% \begin{align}
% \label{omega_def}
% \omega&:=\nabla^\perp\cdot\bu:=\partial_{x_1} u_2-\partial_{x_2} u_1.
% \end{align}
% Taking the (scalar) curl of the momentum equation \eqref{Bouss_mo}, we find
% \begin{align}\label{vor_eqn}
% \partial_t\omega + (\bu\cdot\nabla)\omega &= \nu\triangle\omega+
% \partial_{x_1}\theta.
% \end{align}

% ===============================================
% ========== End Preliminary Section ============
% ===============================================
 
% ==============================================================
\section{Backward Uniqueness}\label{sec_Semigroup}
% ==============================================================
Our main result in this section is  that the system \eqref{Bouss} has the property of backward uniqueness for sufficiently smooth solutions; namely, that if two (sufficiently smooth) solutions agree at some point in time, then they agree for all previous times.   While this is known for the 2D Navier-Stokes equations and for the transport equation, the result for the coupled system \eqref{Bouss} does not seem to follow in a straightforward manner from the corresponding results for the uncoupled system. Our technique is inspired by that  of \cite{Bardos_Tartar_1971}. The backward uniqueness result will be useful in studying 
the asymptotic dynamics of \eqref{Bouss}. For instance, together with well-posedness, it implies that if any two global trajectories intersect, then they must be identical.

Let us recall that the method in \cite{Bardos_Tartar_1971} (which
applies to a wide class of fully dissipative equations, including the 2D
Navier-Stokes equations) is to consider two solutions $\bu_1$ and $\bu_2$ which
are not identical at some time, and then to show that the log of the norm of
their difference, $\bud:=\bu_1-\bu_2$, remains bounded for all future times, so
that the solutions can never coincide in the future (this is what is meant by
backward uniqueness).  The estimates in \cite{Bardos_Tartar_1971}  
typically require one to show that the Dirichlet quotient $\lambda(t):=\|\bud\|_{\h^1}^2/\|\bud\|_{L^2}^2$  grows at most exponentially in time.  
The lack of diffusion in \eqref{Bouss_den} poses a hurdle in showing such a bound in our case.
% since  it is not clear how one should combine the various possible norms of $\bud$ and $\thetad$ to extract a similar result, especially since control of $\|\thetad\|$ at first seems necessary.   
However, we show below that by considering instead the hybrid quantity 
$I(t):=\frac{\|\bud\|_{\h^1}^2}{\|\bud\|_{L^2}^2+g^2\|\thetad\|_{L^2}^2}$, the
diffusion in \eqref{Bouss_mo} alone is enough to gain the necessary control over
both variables $\bu$ and $\theta$.  Note that $I(t)$ is not a Dirichlet quotient.

\begin{theorem}\label{thm_backward_uniqueness}
 For $i=1,2$, let $(\bu_i,\theta_i)$ be solutions to \eqref{Bouss}, such
that $(\bu_i,\theta_i) \in L^\infty((T_1,T_2);H\cap \h^3)\times
L^\infty((T_1,T_2);\h^3)$.  If there exists a time $T^* \in (T_1,T_2)$ such
that $(\bu_1(T^*),\theta_1(T^*)) = (\bu_2(T^*),\theta_2(T^*))$, then 
for a.e. $t\in (T_1,T_2)$, we have $(\bu_1(t),\theta_1(t)) = (\bu_2(t),\theta_2(t))$.
\end{theorem}

\begin{proof}
 Denote $\bud :=\bu_1-\bu_2$, $\thetad:=\theta_1-\theta_2$,
$\widetilde{p}:= p_1-p_2$.  
\begin{subequations}\label{Bouss_eps_diff}
\begin{align}
\label{Bouss_eps_diff_mo}
\partial_t\bud + (\bu_1\cdot\nabla)\bud+ (\bud\cdot\nabla)\bu_2 +\nabla
\widetilde{p}&= \nu\triangle\bud+ \thetad \bg,
\\
\label{Bouss_eps_diff_den}
\partial_t\thetad 
+(\bu_1\cdot\nabla)\thetad +(\bud\cdot\nabla)\theta_2 
 &=0.
\end{align}
\end{subequations}

Suppose that there exists some $T^*>0$ for which both $\bud(T^*)\equiv0$ and
$\thetad(T^*)\equiv0$, but that there also exists some $t_0\in (T_1,T^*)$ such that
$\bud(t_0)\not\equiv0$ or $\thetad(t_0)\not\equiv0$.  Due to the continuity in
time given by Theorem \ref{thm_Well_Posedness2}, there is an interval
$[t_0,t_0+\delta)$ on which $\bud\not\equiv0$ or $\thetad\not\equiv0$.  Let us
denote by $[t_0,t_1)$ the largest such interval, noting that we must have
$\bud(t_1)\equiv0$ and $\thetad(t_1)\equiv0$.  For the remainder of this
section, we work on this interval.  Taking the inner product of
\eqref{Bouss_eps_diff_mo} with $\bud$, and with $A\bud$, and of
\eqref{Bouss_eps_diff_den} with $\thetad$, integrating by parts, and using \eqref{Bzero} yields
\begin{subequations}\label{Bouss_eps_diff_L2}
\begin{align}
\label{Bouss_eps_diff_mo_L2}
\frac12\frac{d}{dt}\|\bud\|_{L^2}^2 &= -\nu\|\bud\|_{\h^1}^2-\ip{(\bud\cdot\nabla)\bu_2}{\bud}+
\ip{\thetad \bg}{\bud},
\\
\label{Bouss_eps_diff_mo_V}
\frac12\frac{d}{dt}\|\bud\|_{\h^1}^2 &= 
-\nu\|A\bud\|_{L^2}^2  -\ip{(\bud\cdot\nabla)\bu_2}{A\bud}+ \ip{\thetad \bg}{A\bud},
\\
\label{Bouss_eps_diff_den_L2}
\frac12\frac{d}{dt}\|\thetad\|_{L^2}^2
 &=-\ip{(\bud\cdot\nabla)\theta_2}{\thetad}.
\end{align}
\end{subequations}

Consider the quantity $L$ defined on $[t_0,t_1)$ by
\begin{align}
 L(t):=\log\pnt{\frac{1}{\sqrt{\|\bud\|_{L^2}^2+g^2\|\thetad\|_{L^2}^2}}} = -\frac{1}{2}\log\pnt{\|\bud\|_{L^2}^2+g^2\|\thetad\|_{L^2}^2}
\end{align}
Using \eqref{Bouss_eps_diff_mo_L2} and \eqref{Bouss_eps_diff_den_L2}, we compute
\begin{align}
\label{dL_eqn}
 \frac{dL}{dt}
 &=\frac{-1}{\|\bud\|_{L^2}^2+g^2\|\thetad\|_{L^2}^2}\pnt{\frac12\frac{d}{dt}\|\bud\|_{L^2}^2+\frac{g^2}{2}\frac{d}{dt}\|\thetad\|_{L^2}^2}
 \\&= \notag
 \frac{\nu\|\bud\|_{\h^1}^2+\ip{(\bud\cdot\nabla)\bu_2}{\bud}- \ip{\thetad \bg}{\bud}+g^2\ip{(\bud\cdot\nabla)\theta_2}{\thetad}}{\|\bud\|_{L^2}^2+g^2\|\thetad\|_{L^2}^2}
 \\&\leq\notag
 \nu \frac{\|\bud\|_{\h^1}^2}{\|\bud\|_{L^2}^2+g^2\|\thetad\|_{L^2}^2}
  +\frac{\ip{(\bud\cdot\nabla)\bu_2}{\bud}}{\|\bud\|_{L^2}^2+g^2\|\thetad\|_{L^2}^2}
  +\frac12
  +\frac{g^2\ip{(\bud\cdot\nabla)\theta_2}{\thetad}}{\|\bud\|_{L^2}^2+g^2\|\thetad\|_{L^2}^2}
  \\&:= \notag
  \nu I+II+\frac12+III
\end{align}
where we have used Young's inequality. Thus, in order to find a bound for $dL/dt$, we search for a bound on the quotients $I$, $II$, and $III$ in the right-hand side of \eqref{dL_eqn}.  First, note that, due to \eqref{B424} and \eqref{Bi22}, 
\begin{align}
 \ip{(\bud\cdot\nabla)\bu_2}{\bud}
 &\leq 
 C\|\bu_2\|_{\h^1}\|\bud\|_{L^2}\|\bud\|_{\h^1}
 \leq 
 C\|\bu_2\|_{\h^1}^2\|\bud\|_{L^2}^2+\|\bud\|_{\h^1}^2,
\\
\label{u_theta2_thetad_ineq}
 g^2\ip{(\bud\cdot\nabla)\theta_2}{\thetad}
 &\leq
 Cg^2\|\triangle\theta_2\|_{L^2}\|\bud\|_{\h^1}\|\thetad\|_{L^2}
 \leq 
 Cg^4\|\triangle\theta_2\|_{L^2}^2\|\thetad\|_{L^2}^2+\|\bud\|_{\h^1}^2.
%  \\
%   g^2\ip{(\bud\cdot\nabla)\theta_2}{\thetad}
%  &\leq
%  Cg^2\|\theta_2\|\|\bud\|_{L^2}^{1/2}\|A\bud\|_{L^2}^{1/2}\|\thetad\|_{L^2}
%  \leq 
%  Cg^4\|\theta_2\|^2\|\thetad\|_{L^2}^2+\|\bud\|_{L^2}\|A\bud\|_{L^2}.
\end{align}
Thus,
\begin{align}
\label{II_bound}
 II &:= \frac{\ip{(\bud\cdot\nabla)\bu_2}{\bud}}{\|\bud\|_{L^2}^2+g^2\|\thetad\|_{L^2}^2}
 \leq
 \frac{ C\|\bu_2\|_{\h^1}^2\|\bud\|_{L^2}^2}{\|\bud\|_{L^2}^2+g^2\|\thetad\|_{L^2}^2}
 +
 \frac{\|\bud\|_{\h^1}^2}{\|\bud\|_{L^2}^2+g^2\|\thetad\|_{L^2}^2}
 \leq
  C\|\bu_2\|_{\h^1}^2
 +
 I,
\end{align}
and
\begin{align}
\label{III_bound}
 III &:= 
 \frac{g^2\ip{(\bud\cdot\nabla)\theta_2}{\thetad}}{\|\bud\|_{L^2}^2+g^2\|\thetad\|_{L^2}^2}
 \leq
 \frac{Cg^4\|\triangle\theta_2\|_{L^2}^2\|\thetad\|_{L^2}^2+\|\bud\|_{\h^1}^2}{\|\bud\|_{L^2}^2+g^2\|\thetad\|_{L^2}^2}
 \leq
  Cg^2\|\triangle \theta_2\|_{L^2}^2+I
\end{align}
% \begin{align}
%  III &:= 
%  \frac{g^2((\bud\cdot\nabla)\theta_2,\thetad)}{\|\bud\|_{L^2}^2+g^2\|\thetad\|_{L^2}^2}
%  \leq
%  \frac{Cg^4\|\theta_2\|^2\|\thetad\|_{L^2}^2+\|\bud\|_{L^2}\|A\bud\|_{L^2}}{\|\bud\|_{L^2}^2+g^2\|\thetad\|_{L^2}^2}
%  \leq\notag
%   Cg^2\|\theta_2\|^2+\frac{\|\bud\|_{L^2}\|A\bud\|_{L^2}}{\|\bud\|_{L^2}^2+g^2\|\thetad\|_{L^2}^2}
% \end{align}

In order to estimate $I$, we compute $dI/dt$.  \comments{A similar approach was first employed in \cite{Bardos_Tartar_1971} in the context of the 2D Navier-Stokes equations, which formally corresponds to setting $g=0$ in $I$.}  First, we derive an inequality which will be useful in estimating $dI/dt$.  
%We note that there is a  difference from the corresponding computation for the 2D Navier-Stokes equations in \cite{Bardos_Tartar_1971}; namely that we have an inequality rather than an identity.
\begin{align}
\label{viscous_term_bound}
&\quad
 \|A\bud\|_{L^2}^2  -I\|\bud\|_{\h^1}^2
 =
 \|A\bud\|_{L^2}^2  -\frac{\|\bud\|_{\h^1}^4}{\|\bud\|_{L^2}^2+g^2\|\thetad\|_{L^2}^2}
   \\&\geq\notag
 \|A\bud\|_{L^2}^2  -2\frac{\|\bud\|_{\h^1}^4}{\|\bud\|_{L^2}^2+g^2\|\thetad\|_{L^2}^2}
 +\frac{\|\bud\|_{\h^1}^4}{\|\bud\|_{L^2}^2+g^2\|\thetad\|_{L^2}^2}
 \cdot
 \frac{\|\bud\|_{L^2}^2}{\|\bud\|_{L^2}^2+g^2\|\thetad\|_{L^2}^2}
 \\&=\notag
 \|A\bud\|_{L^2}^2  -2I\ip{A\bud}{\bud}+I^2\|\bud\|_{L^2}^2
 \\&=\notag
 \|A\bud  -I\bud\|_{L^2}^2.
\end{align}
Using \eqref{Bouss_eps_diff_L2}, we compute
{\allowdisplaybreaks
\begin{align}
 \frac12\frac{dI}{dt}
 &:=
 \frac{d}{dt}\frac{\frac12\|\bud\|_{\h^1}^2}{\|\bud\|_{L^2}^2+g^2\|\thetad\|_{L^2}^2}
 = 
 \frac{\frac12\frac{d}{dt}\|\bud\|_{\h^1}^2
 }{\|\bud\|_{L^2}^2+g^2\|\thetad\|_{L^2}^2}
 -
  \frac{\|\bud\|_{\h^1}^2\frac12\frac{d}{dt}(\|\bud\|_{L^2}^2+g^2\|\thetad\|_{L^2}^2)
 }{(\|\bud\|_{L^2}^2+g^2\|\thetad\|_{L^2}^2)^2}
 \\\notag&= 
 \frac{-\nu\|A\bud\|_{L^2}^2  -\ip{(\bud\cdot\nabla)\bu_2}{A\bud}+ \ip{\thetad \bg}{A\bud}
 }{\|\bud\|_{L^2}^2+g^2\|\thetad\|_{L^2}^2}
 \\\notag&\notag\quad+
  \frac{\|\bud\|_{\h^1}^2\pnt{
  \nu\|\bud\|_{\h^1}^2+\ip{(\bud\cdot\nabla)\bu_2}{\bud}- \ip{\thetad \bg}{\bud}
  +g^2\ip{(\bud\cdot\nabla)\theta_2}{\thetad}
  }
 }{(\|\bud\|_{L^2}^2+g^2\|\thetad\|_{L^2}^2)^2}
  \\\notag&= 
 \frac{-\nu\|A\bud\|_{L^2}^2  -\ip{(\bud\cdot\nabla)\bu_2}{A\bud}+ \ip{\thetad \bg}{A\bud}
 }{\|\bud\|_{L^2}^2+g^2\|\thetad\|_{L^2}^2}
 \\&\notag\quad+
  I\cdot\frac{
  \nu\|\bud\|_{\h^1}^2+\ip{(\bud\cdot\nabla)\bu_2}{\bud}
  -\ip{\thetad \bg}{\bud}
  +g^2\ip{(\bud\cdot\nabla)\theta_2}{\thetad}
 }{\|\bud\|_{L^2}^2+g^2\|\thetad\|_{L^2}^2}
   \\\notag&= 
 \frac{-\nu(\|A\bud\|_{L^2}^2  -I\|\bud\|_{\h^1}^2)
 -\ip{(\bud\cdot\nabla)\bu_2}{A\bud-I\bud}
 + \ip{\thetad \bg}{A\bud-I\bud}
 }{\|\bud\|_{L^2}^2+g^2\|\thetad\|_{L^2}^2}
 \\&\notag\quad+
  I\frac{
  g^2\ip{(\bud\cdot\nabla)\theta_2}{\thetad}
 }{\|\bud\|_{L^2}^2+g^2\|\thetad\|_{L^2}^2}.
 \end{align}
} % end \allowdisplaybreaks
Using this with \eqref{B442}, \eqref{viscous_term_bound} and \eqref{B2i2}, we find,
 {\allowdisplaybreaks
\begin{align}
    \frac12\frac{dI}{dt}&\leq
 \frac{-\nu\|A\bud  -I\bud\|_{L^2}^2
 +\|\bud\|_{\h^1}\|A\bu_2\|_{L^2}\|A\bud  -I\bud\|_{L^2} %442
 + g\|\thetad\|_{L^2}\|A\bud  -I\bud\|_{L^2}
 }{\|\bud\|_{L^2}^2+g^2\|\thetad\|_{L^2}^2}
 \\&\notag\quad+
 I\frac{
  g^2\|\bud\|_{L^2}\|\theta_2\|_{\h^3}\|\thetad\|_{L^2}
 }{\|\bud\|_{L^2}^2+g^2\|\thetad\|_{L^2}^2}
     \\\notag&\leq
 \frac{-\nu\|A\bud  -I\bud\|_{L^2}^2/2
 +\|\bud\|_{\h^1}^2\|A\bu_2\|_{L^2}^2/\nu
 + g^2\|\thetad\|_{L^2}^2/\nu 
 }{\|\bud\|_{L^2}^2+g^2\|\thetad\|_{L^2}^2}
+ I\frac{g\|\theta_2\|_{\h^3}}{2}
     \\\notag&\leq
 I\frac{\|A\bu_2\|_{L^2}^2}{\nu}
  + \frac{1}{\nu}
+ I\frac{g\|\theta_2\|_{\h^3}}{2},
\end{align}
} % end \allowdisplaybreaks
where we have used Young's inequality for the last estimate.  Applying Gr\"onwall's inequality on $[t_0,t)$, $t\in[t_0,t_1)$ now yields
\begin{align}
 I(t)
&\leq 
 I(0) e^{\int_0^tK(s)\,ds}+\frac{1}{\nu}\int_0^t e^{\int_\tau^tK(s)\,ds}\,d\tau
 \\&\leq \notag
 I(0) e^{\int_0^{T^*}K(s)\,ds}+\frac{1}{\nu}\int_0^{T^*} e^{\int_\tau^{T^*}K(s)\,ds}\,d\tau,
\end{align}
where $K(s):= \frac{\|A\bu_2(s)\|_{L^2}^2}{\nu} + \frac{g\|\theta_2(s)\|_{\h^3}}{2}$
Thus, $I$ is bounded on $[t_0,t)$.  By \eqref{II_bound} and \eqref{III_bound}, this in turn shows that $II$ and $III$ are also bounded, so that \eqref{dL_eqn} implies that $dL/dt$ is bounded on $[t_0,t)$. Integrating in $dL/dt$ in $t$ over $[t_0,T^*)$, we find that $L$ is bounded on $[t_0,T^*)$.  In particular, 
\begin{align}
 \pnt{\|\bud\|_{L^2}^2+g^2\|\thetad\|_{L^2}^2}^{-1/2}
\end{align}
is bounded on $[t_0,T^*)$, contradicting the assumptions that $\bud(T^*)=0$ and $\thetad(T^*)=0$.
\end{proof}

 The global well-posedness results (Theorem \ref{thm_Danchin_Paicu}--\ref{thm_Well_Posedness2})
 imply that \eqref{Bouss} has a well-defined semigroup  operator $S(t)$ defined for $t\geq0$ by 
\begin{align}\label{def_semigroup_operator}
S(t)(\buInit,\thetaInit) = (\bu(t),\theta(t)).
\end{align}
The backward uniqueness result implies that the semigroup is injective on at least the smooth portion of the set which we will define to be the attractor.  This in turn implies that $S(t)$ can be extended to hold for negative times on the smooth trajectories in that set.

% ===============================================
% ======= End Backward Uniqueness Section =======
% ===============================================

% ===============================================
% =========== End Semigroup Section =============
% ===============================================
 
 % =====================================================================
\section{An Adequate Notion of an Attractor}\label{sec_notion_of_attractor}
% =====================================================================
We look for a notion of an attractor which applies to
\eqref{Bouss}.  While some analogies with the attractor, $\cA_{\text{NS}}$, of the 2D
Navier-Stokes equations can be made, there are also striking differences.  Therefore, for comparison, let us recall that for the two-dimensional Navier-Stokes
equations,
\begin{align}\label{NSE}
 \frac{d\bu}{dt}+\nu A\bu + B(\bu,\bu) = f
\end{align}
(where $f\in H$ is time-independent), the global attractor $\cA_{\text{NS}}$ has
the following equivalent definitions (see, e.g., \cite{Constantin_Foias_1988,Robinson_2001}):
\begin{enumerate}[(I)]
 \item\label{A_attracting} $\cA_{\text{NS}}$ is the minimal set which uniformly
attracts all bounded sets in $H$ as $t\maps\infty$.
 \item\label{A_bounded}$\cA_{\text{NS}}$ is the maximal compact set of all $\buInit\in H$ such that
the solution  of \eqref{NSE} satisfying $\bu(0)=\buInit$ exists for all $t\in\nR$
(i.e., backward and forwards in time) and satisfies
 \begin{align}
  \sup_{t\in\nR}\|\bu(t)\|_{L^2}<\infty
 \end{align}
 \item\label{A_invariant} $\cA_{\text{NS}}$ is the maximal bounded set $X\subset
H$ such that $S_{\text{NS}}(t)X=X$ for all $t\in\nR$, where $S_{\text{NS}}(t)$ is the semigroup operator
for \eqref{NSE}.
\end{enumerate}
 Moreover, it is well-known that the attractor for the two-dimensional Navier-Stokes equations has finite fractal dimension (and therefore finite Hausdorff dimension) (see, e.g., \cite{Constantin_Foias_1988, CFT85}). 

Clearly, any reasonable notion of attractor of a system must include the steady
states of the system.   While the above definitions are equivalent ways to define the global
attractors for many dissipative systems, this is not the
case for system \eqref{Bouss} due to the lack of dissipation in
\eqref{Bouss_den}.  Indeed, we will describe below (see Remark \ref{rem:steadystates}) a set of steady states which are neither bounded nor 
finite-dimensional. Next we describe special classes of solutions of \eqref{Bouss} which play an important role in our study.
% (even in terms of Hausdorff or fractal dimensions or in the sense of existence of finite number determining modes or nodes). Moreover, noting that 
% the set $\{U_{\thetaInit}:\|\thetaInit\| \le r\}$ is not compact in (the norm topology of) $X$, 
% any such attractor, even when intersected with a bounded set in $X$, cannot be compact.
% In this section, we will define an appropriate analogue of
% the global attractor for \eqref{Bouss} using only property \eqref{A_bounded}
% and then discuss in what sense it is ``attracting".
%However, we will also study weaker forms of \eqref{A_attracting} and
%\eqref{A_invariant}.

\subsection{Special Classes of Solutions and Steady States}  \label{steadystates}%---------------------------------------------
%We establish here that the attractor for \eqref{Bouss} is infinite-dimensional. 
%This immediately implies that property \ref{} for  $\cA_{text{NS}}$ does not hold.
Here we will provide three families of explicit solutions of \eqref{Bouss} corresponding to purely vertical (V),  purely horizontal (H) and plane wave solutions.

\begin{itemize}
\item[(i)] {\bf Vertical solutions: }
Let $a^V,\theta^V\in L^2([0,L])$ be an arbitrary periodic, mean-zero function depending only on the vertical variable: $\theta^V=\theta^V(x_2)$.  Set $u^V_2=0$, and define $u^V_1=u^V_1(x_2,t)$ by the following unforced diffusion problem:
\begin{equation}\label{vertical_heat_eqn}
\begin{cases}
  \partial_t u^V_1 -\nu\partial_{x_2}^2 u^V_1 &= 0,
 \\
 u^V_1(x_2,0) &= a^V(x_2),
\end{cases}
\end{equation}
along with periodic boundary conditions on $[0,L]$ and the mean-free condition.  Let $p^V$ be defined up to an arbitrary constant by\footnote{Property \eqref{p_V_def} is often referred to as hydrostatic balance in the geophysical literature.}
\begin{align}\label{p_V_def}
 \partial_{x_2}p^V = g\theta^V.
\end{align}
It is easy to check that $(u_1,u_2,\theta,p)=((u^V_1,u^V_2),\theta^V,p^V)$ is a solution to \eqref{Bouss_mo}-\eqref{Bouss_div} with initial data $u_1(\bx,0)=a^V(x_2)$, $u_2(\bx,0) = 0$, and $\theta(\bx,0)=\theta^V$.

\noindent
{\bf Steady State:} It is easy to see that the vertical solution $(\bu^V,\theta^V)$ defined above converges to the steady state $({\mathbf 0},\theta^V)$ as $t \ra \infty$.

\item[(ii)] {\bf Horizontal Solutions:}  Let $a^H,\theta^H\in L^2([0,L])$ be  arbitrary periodic, mean-zero functions depending only on the horizontal variable: $\theta^H=\theta^H(x_1)$.  Set $u^H_1=0$, $p^H=0$, and let $u^H_2=u^H_2(x_1,t)$ be the unique solution of the following forced diffusion problem:
\begin{equation}\label{horizontal_heat_eqn}
\begin{cases}
  \partial_t u^H_2 -\nu\partial_{x_1}^2 u^H_2 &= g\theta^H,
 \\
 u^H_2(x_1,0) &= a^H(x_1),
\end{cases}
\end{equation}
 with periodic boundary condition on $[0,L]$ and the mean-free condition.  
It is easy to check that 
\(
(\bu^H,\theta^H,p^H):=((u^H_1,u^H_2),\theta^H,p^H)
\)
 is a solution to \eqref{Bouss_mo}-\eqref{Bouss_div} with initial data 
\begin{align}\label{horizontal_ic}
 u_1(\bx,0)=0, u_2(\bx,0) = a^H(x_1)\ \mbox{and}\  \theta(\bx,0)=\theta^H.
\end{align}

Since these flows are independent of $x_2$, they can be thought of in a geophysical context as \textit{columnular flows}.  In Section \eqref{sec_Energy_Enstrophy}, we show that the steady states corresponding to these flows also arise as solutions lying on the boundary of a certain set containing the attractor after projecting into the Energy-Enstrophy plane.

\noindent
{\bf Steady State:} Let $u_2=u_2(x_1)$ be the unique, mean-free periodic solution of the equation
\begin{align}
\nu \frac{d^2}{dx_1^2}u_2(x_1)= - g \theta^H(x_1).
\end{align}
One can check that the horizontal solution $(\bu^H,\theta^H)$ defined above converges, as $t \ra \infty$,  to the steady state
\begin{align}
 u_1=0, u_2, \theta= \theta^H.
\end{align}

\item[(iii)] {\bf Plane Wave solutions:} We also consider certain plane-wave (PW) solutions, first introduced in the context of the Navier-Stokes equations in \cite{FS}.  
Specifically, let $f = f(z,t)$ and $h = h(z)$ be smooth mean-free functions\footnote{We may let $h=h(z,t)$ \textit{a priori}, but our ansatz for the form of solutions will force $\partial_t h \equiv0$, so we assume that $h$ is time-independant from the outset.} which are $L-$periodic in $z$, and fix a non-zero vector $\bk=(k_1,k_2)\in\nZ^2$ where $k_1+k_2=0$. Below, in \eqref{f_heat_equation}, we will require $f$ to satisfy a particular heat equation.  Define $\bu$ and $\theta$ by 
\begin{subequations}\label{plane_wave_ansatz}
\begin{align}
\label{plane_wave_ansatz_u}
 \bu^{\text{PW}}(\bx,t)&:=(f(\bk\cdot \bx,t),f(\bk\cdot \bx,t)), 
 \\\label{plane_wave_ansatz_theta}
 \theta^{\text{PW}}(\bx,t)&:=h(\bk\cdot \bx).
\end{align}
\end{subequations}
Note  that
\begin{align*}
 \nabla\bu^{\text{PW}} = 
 \begin{pmatrix} k_1 & k_2 \\ k_1 & k_2
     \end{pmatrix}f_z(\bk\cdot \bx,t),
     \
     \nabla\theta^{\text{PW}} = (k_1,k_2)h_z(\bk\cdot \bx).
\end{align*}
This implies $\nabla\cdot\bu^{\text{PW}} \equiv\text{tr}(\nabla\bu^{\text{PW}}) = (k_1+k_2)f_z = 0$, i.e.,
$\bu^{\text{PW}}$ is divergence free. Furthermore,  
\begin{align}
& (\bu^{\text{PW}}\cdot\nabla)\bu^{\text{PW}} = (1,1)(k_1+k_2)ff_z = (0,0)\  \mbox{ and} 
\\\ \notag
& (\bu^{\text{PW}}\cdot\nabla)\theta^{\text{PW}} = (k_1+k_2)fh_z = 0.  
\end{align}
Assume that $\bu^{\text{PW}}$ and $\theta^{\text{PW}}$ defined as in \eqref{plane_wave_ansatz} give a solution for some pressure $p$.  Substituting these relations into \eqref{Bouss_mo} and 
\eqref{Bouss_den}, with $z = \bk \cdot \bx$ and $f(x_1,x_2,t)=f(z,t)$, we obtain 
\begin{subequations}\label{plane_wave}
\begin{align}
\label{plane_wave_mo_1}
\partial_tf + p_{x_1}&= \nu|\bk|^2f_{zz},
\\
\label{plane_wave_mo_2}
\partial_tf + p_{x_2}&= \nu|\bk|^2f_{zz}+ gh.
% \\
% \label{plane_wave_den}
% \partial_th \;\;\,\quad\quad &=0.
\end{align}
\end{subequations}
Applying $\partial_{x_2}$ to \eqref{plane_wave_mo_1} and $\partial_{x_1}$ to \eqref{plane_wave_mo_2} and subtracting the results, we obtain
\begin{align}\label{curl_f_equation}
 (k_1-k_2)\partial_tf_z &= \nu|\bk|^2(k_1-k_2)f_{zzz}+ k_1gh_z,
\end{align}
Since $k_1+k_2=0$, it follows that $k_1-k_2 = 2k_1$.  Dividing \eqref{curl_f_equation} by $2k_1$ and integrating in $z$ implies
 $\partial_tf = \nu|\bk|^2f_{zz}+ \tfrac{1}{2}gh+\psi$, for some undetermined $\psi=\psi(t)$.
Since  $f$ and $h$ are mean-free and periodic, we conclude that $\psi\equiv0$.  Thus, $f$ satisfies the heat equation
\begin{align}\label{f_heat_equation}
 \partial_tf &= \nu|\bk|^2f_{zz}+ \tfrac{1}{2}gh.
\end{align}

As usual, we can solve for the pressure by applying $\partial_{x_1}$ to \eqref{plane_wave_mo_1} and $\partial_{x_2}$ to \eqref{plane_wave_mo_2} and adding the results to obtain
\begin{align}
 \triangle p (\bx,t)= gk_2h_z(\bk\cdot \bx,t).
\end{align}
This, along with the mean-free and periodic boundary conditions, defines the pressure uniquely.  In fact, we can get a more explicit formula for the pressure gradient by subtracting \eqref{plane_wave_mo_2} from \eqref{plane_wave_mo_1} to obtain
$p_{x_2}-p_{x_1}= gh$.
Adding \eqref{plane_wave_mo_1} to \eqref{plane_wave_mo_2} and comparing with \eqref{f_heat_equation} yields $p_{x_1}+p_{x_2}=0$.
Thus, 
\begin{align}\label{f_pressure}
p_{x_1} =-\tfrac{1}{2}gh
\quad\text{and}\quad
p_{x_2} =\tfrac{1}{2}gh.
\end{align}
Conversely, let $h=h(z)$ be any mean free, periodic function and let 
 $f$ satisfy  \eqref{f_heat_equation}. Define $(\bu,\theta)$ by \eqref{plane_wave_ansatz} and the pressure by
 \begin{gather}  \label{pressure_def}
 p\equiv p^{\text{PW}}:= - \frac{1}{2k_1}gH(\bk \cdot \bx),\  \mbox{ where}\  H'(z)=h(z).
 \end{gather}
 It is easy to see that $(\bu^{\text{PW}},\theta^{\text{PW}})$ is a solution of \eqref{Bouss}.
 We note that the solutions of the form $(\bu^{\text{PW}},\theta^{\text{PW}})$ thus defined 
  are distinct from $(\bu^V,\theta^V)$ and $(\bu^H,\theta^H)$, e.g., since both components of $\bu^{\text{PW}}$ are non-zero. Moreover, they are steady states if $f$ is time-independent and is given by
  $f=f_{steady}$, where $f_{steady}=f_{steady}(z)$ is the unique, mean-free periodic solution of the equation 
  \begin{gather}  \label{fsteady}
  \frac{d^2}{dz^2} f_{steady}(z)= -\frac{g}{2\nu |\bk|^2} h(z).
  \end{gather}
  
  \noindent
  {\bf Steady State:} Due to \eqref{pressure_def} and \eqref{f_heat_equation}, it is easy to see that
  $(\bu^{\text{PW}},\theta^{\text{PW}})$ approach a steady state which is a plane-wave solution corresponding to the same $h$ and $f=f_{steady}$ where $f_{steady}$ is as in \eqref{fsteady}.

\end{itemize}

\begin{remark}  \label{rem:steadystates}
  The solution families $(\bu^V,\theta^V)$, $(\bu^H,\theta^H)$ and $(\bu^{\text{PW}},\theta^{\text{PW}})$ given above have unique steady states which they approach exponentially fast in $L^2$. Since these steady states are determined by $\theta^V$, $\theta^H$, and $h$, respectively, which are arbitrary smooth, mean-free, periodic functions, this show that the set of steady states of system \eqref{Bouss} is fairly rich, and in particular, it is both infinite-dimensional (in the sense that it contains infinite dimensional subspaces) and unbounded. 
\end{remark}

\subsection{Weak Sigma-Attractor} \label{Weak Sigma-Attractor}

\begin{definition}\label{def_attractor}
% Possible names: conservative set, nexus, binding set, eternal set
 Let $S(t)$ be the semi-group operator associated with \eqref{Bouss} via \eqref{def_semigroup_operator}.  We define the \textit{weak sigma-attractor $\cA$, of the semi-dissipative system} \eqref{Bouss} to be the set of all $(\buInit,\thetaInit)\in H\times L^2$ with the property that
 \begin{enumerate}[(i)]
   \item There exists a global trajectory $(\bu(t),\theta(t))$ defined for all $ t \in \nR$ such that $(\bu(t),\theta(t))$ 
   belongs to $H \times L^2$ and solves  \eqref{Bouss} for all $t \in \nR$
     and moreover $(\bu(0),\theta(0))=(\buInit,\theta_0)$.
  \item The trajectory $(\bu(t),\theta(t))$ is globally bounded, i.e., the set
  $\{(\bu(t),\theta(t)): t \in \nR\}$ is  bounded  in $H \times L^2$.
 \end{enumerate}
\end{definition}
Clearly, it follows from the definition that every bounded, global trajectory $\{(\bu(t),\theta(t)): t \in \nR\}$
lies entirely on the attractor. The justification for referring to $\cA$ as the weak-sigma attractor will be provided towards the end of this section. We will show that $\cA$ is a non-empty, proper subset of the phase space.

\comments{
\begin{remark}
 Note that in Definition \ref{def_attractor}, we only require the existence of a trajectory which is bounded backward and forward in time.  We do not require it to be unique, and therefore, we are not relying on backward uniqueness here.  However, we see that Theorem \ref{thm_backward_uniqueness} implies that $S(t)$ restricted to $\cA\cap(\h^3\times \h^3)$ forms a dynamical system evolving in $\nR$.
\end{remark}
}

We will now justify our terminology, namely that of the weak sigma-attractor, for $\cA$.  The global attractor for a dissipative system is often defined as the
$\omega$-limit set of the compact absorbing ball.  However, system \eqref{Bouss}
does not have a compact absorbing ball.  Yet, since $\|\theta(t)\|_{L^2}=\|\thetaInit\|_{L^2}$ is
time-invariant, equation \eqref{Bouss_mo}, considered alone, has a compact
absorbing ball in $H$ by the standard theory of the 2D Navier-Stokes equations
(see Appendix \ref{ssec_Grashof_Numbers}).  By varying the value of $\|\thetaInit\|_{L^2}$,
one obtains a family of absorbing balls.  This observation is exploited in Section \ref{sec_main_proof} to write $\cA$ as the union of $\omega$-limit sets.  

Given initial data $(\buInit,\thetaInit)\in H\times L^2$, define the dimensionless, time-independent 
Grashof-type number for \eqref{Bouss} to be
\begin{align}\label{def_L2_Grashof}
 G := 
 \frac{g\|\thetaInit\|_{L^2}}{\nu^2\kappa_0^2}.
\end{align}
Standard energy estimates yield
\begin{align}\label{L2_energy_estimate}
 \|\bu(t)\|_{L^2}^2
 \leq
 e^{-\nu\kappa_0^2t}\|\buInit\|_{L^2}^2
 +\nu^2G^2(1-e^{
-\nu\kappa_0^2t}).
\end{align}
Thus, there exists a time $t_*=t_*(\|\buInit\|_{L^2})$ such that, for $t>t_*$, 
$\|\bu(t)\|_{L^2}$ is in the ball of radius $2\nu G$ in $H$.  For example, $t_*$
can be taken as
\begin{gather}  \label{inittimeset}
 t_*(\|\buInit\|_{L^2}) = \frac{1}{\nu\kappa_0^2}
 \max\set{1,\log\frac{\|\buInit\|_{L^2}^2}{3\nu^2G^2}}.
\end{gather}

\begin{definition}\label{def_B_r}
Let us denote the following Cartesian product of balls: 
\begin{align*}
 B_r :=\set{(\buInit,\thetaInit)\in H\times L^2 :  \|\thetaInit\|_{L^2}\leq r, \|\buInit\|_{L^2}\leq R(r)},
 \qquad
 R(r):= 2\frac{gr}{\nu\kappa_0^2}.
\end{align*}
\end{definition}
Due to  \eqref{Lp_conserved} and \eqref{L2_energy_estimate}, 
for a fixed $r \geq 0$,  the set $B_r$ is semi-invariant for all positive times, i.e., 
\begin{align}
S(t)B_r\, \subset B_r\ \mbox{ for all}\  t>0, 
\end{align}
where $S(t)$ is the solution semigroup 
for \eqref{Bouss} defined in \eqref{def_semigroup_operator}. 
% With this observation in hand, we are now ready to make a definition that will be key to our discussion.
\begin{definition}  The \emph{local attractor at level $r$}, denoted by $\cA_r$,  is defined to be the $\omega$-limit set of $B_r$; that is,
\begin{align} \label{ardef}
 \cA_r
 :=
 \omega(B_r)
 :=
 \bigcap_{\tau\geq0}{}^{\text{wk}}\!\overline{\set{S(t)(\buInit,\thetaInit):t>\tau, (\buInit,\thetaInit)\in B_r}},
 \end{align}
 where the closure is taken in the weak topology of $H\times L^2$. Note that by the Banach-Alaoglu Theorem, $\cA_r$ is weakly compact in $H \times L^2$. Furthermore, by Proposition \ref{wkcontinuity} (in the next section), for each fixed $t$, the map   $S(t):B_r \lra B_r$ is weakly continuous (i.e., continuous with respect to the weak topology on $B_r$). Thus,
\begin{align}
 {}^{\text{wk}}\!\overline{\set{S(t)(\buInit,\thetaInit):t>\tau, (\buInit,\thetaInit)\in B_r}}=S(t)B_r.
\end{align}
\end{definition}
We will show later that the weak sigma-attractor $\cA$ (defined in Definition \ref{def_attractor} as the set of trajectories which are uniformly bounded in $H\times L^2$ for all $t\in\nR$) is indeed equal to the union of the sets $\cA_r$, i.e., 
\begin{align}
\cA = \bigcup_{r \geq0} \cA_r=  \bigcup_{r \geq0, r \in \nQ} \cA_r,
\end{align}
where the last equality follows from the fact that the sets $\cA_r$ are increasing in $r$.
Since $\cA_r$ is weakly compact  in $H \times L^2$, and $\cA$ is a countable union of  weakly compact sets $\cA_r$, this justifies our referring to $\cA$ as the weak sigma-attractor.  We will show later  that in fact
$\cA$ (weakly) attracts all bounded sets.

Let us recall that the weak topology on a separable Hilbert space is metrizable on bounded sets. In fact,  one can define a metric $d$ which is independent of the bounded set so that the corresponding metric space topology coincides with the weak topology on every bounded set. This metric can be defined as follows. 
%Any bounded set is contained in a closed unit disc $D_R$ for sufficiently large $R>0$.
%We may define a metric on  $D_R$ as follows.
 Fix $\{\phi_j\}_{j=1}^\infty$,  a countable dense subset of  the aforementioned Hilbert space, 
and let $X$ be a bounded set in it. 
Define the metric $d$ on $X$ by
\begin{gather}  \label{metricdef}
d(\phi, \psi):= \sum_{j=1}^{\infty} \frac{1}{2^j\|\phi_j\|} |\ip{\phi-\psi}{\phi_j}|,\text{ for all } \phi, \psi \in X.
\end{gather}
This metric space topology coincides with the weak topology on $X$. Henceforth, $d$ will represent this metric on any bounded set. Due to this fact, namely that on bounded sets, the weak topology coincides with the one given by the metric \eqref{metricdef}, we have
\begin{align}
& \cA_r= \{(\bu_a,\theta_a):\ \exists\ t_n \ra \infty\ \mbox{and}\ (\bu_{0,n},\theta_{0,n}) \in B_r \nn \\
&\qquad \qquad \qquad \qquad  
\ \text{ such that } S(t_n)(\bu_{0,n},\theta_{0,n}) \overset{\text{wk}}{\lra }(\bu_a,\theta_a)\}. \label{omegaset}
\end{align}

\begin{remark}
In view of the backward uniqueness of sufficiently smooth trajectories, 
it is natural to ask whether there exists a global attractor in a stronger space, such as in $V\times L^2$, or $V\times H^1$, or another space where global well-posedness holds (cf. Theorems \ref{thm_Well_Posedness1} and \ref{thm_Well_Posedness2}).  In the case of the 2D Navier-Stokes equations with time independent force in $H$, it is known that the global attractor in the phase spaces $V$ and $H$ coincide and moreover, it is contained in $D(A)$ \cite{Constantin_Foias_1988, Robinson_2001}.
However, unlike in the case of the 2D Navier-Stokes equations, it is not known whether one can bound $\bu$ in $\mathcal{D}(A)$ uniformly in time.  Without such a bound, we are not able to prove the existence of a stronger global attractor than the one described here. 
\end{remark}

\section{Structure of the Attractor }\label{sec_main_proof}
In this section, we will state and prove the main result describing the properties of the attractor. 
We will  need the following definition for the semi-distance between two bounded sets in the weak topology. For two bounded sets $A,B $, we define
the semi-distance
\begin{align}
\text{dist}(A,B):=\sup_{x \in A}\ \inf_{y \in B} d(x,y).
\end{align}
It is easy to check that for two bounded sets, $\text{dist}(A,B)=0$ if and only if $A \subset B$.

We are now ready to state our main theorem summarizing the properties of the weak-sigma attractor attractor defined in Subsection \ref{Weak Sigma-Attractor}.
\begin{theorem}  \label{main_att_thm}
The global attractor $\cA$ defined in Definition \ref{def_attractor} has the following properties.
\begin{itemize}
\item[(i)] The relation $\cA = \bigcup\limits_{r\geq0} \cA_r= \bigcup\limits_{r\geq0, r \in \nQ} \cA_r$ holds, where $\cA_r$ is defined by \eqref{ardef}.
\item[(ii)] $\cA$ contains all the steady states of \eqref{Bouss}.
\item[(iii)] $\cA$ has empty interior in the strong (and therefore weak) topology of $H \times L^2$.
\item[(iv)]  The weak sigma-attractor $\cA$ of the system \eqref{Bouss} is  a  non-empty, proper subset  of the phase space $X:=H\times L^2(\Omega)$ which moreover contains infinite-dimensional subspaces of the phase space.
\item[(v)] The set $\cA$ is $\sigma$-compact. More precisely, $\cA$  is a countable union of weakly compact sets $\cA_r, r \in \nQ$.
\item[(vi)] The set $\cA$ attracts all bounded sets. More precisely, if $X \subset H\times L^2$ is a bounded set, then there exists $r >0$ such that ${\it dist}(S(t)X,\cA_r) \ra 0$ as $t \ra \infty $.
\item[(vii)]  For each $r\geq0$, the set $\cA_r$, as well as the set $\cA$, is weakly connected.
\item[(viii)] The global attractor $\cA$ is invariant (i.e. $S(t)\cA =\cA$) for all $t \ge 0$. 
Moreover, it is the minimal set that attracts all bounded sets.
\item[(ix)] (Tracking property) Let $(\bu(t),\theta(t)), t \ge 0$ be a trajectory in $H \times L^2$. There exists a global trajectory $(\bu_\infty(s),\theta_\infty(s)), s \in \nR$ included in $\cA$, with the property that
given  $\epsilon , M>0$, there exists $T>0$ satisfying
\begin{align}
\sup_{s \in [-M,M]}\left\{\|\bu(s+T) - \bu_\infty(s)\|_{L^2} + \|\theta(s+T)-\theta_\infty(s)\|_{{\mathbb H}^{-1}}\right\} < \epsilon.
\end{align}
\end{itemize}
\end{theorem}

Before proceeding to prove the main theorem,  we need to establish the following crucial lemma. 
\subsection{Auxiliary Lemma} %---------------------------------------------
 \begin{lemma}\label{lem_infty_solutions}
 Let $(\bu_{0,n},\theta_{0,n})\in B_r$ and $t_n \ra \infty$ be such that 
 \begin{gather}   \label{omegaconv}
 S(t_n)(\bu_{0,n},\theta_{0,n}) \overset{\text{wk}}{\lra} (\bu_a,\theta_a) \ \in  H\times L^2.
 \end{gather}
 Then $(\bu_a,\theta_a)\in\cA$.  More precisely, there exists a complete bounded trajectory $(\bu_\infty(t),\theta_\infty(t))\in B_r$ for all $t\in\nR$, such that 
\begin{align}
 (\bu_\infty(0),\theta_\infty(0))=(\bu_a,\theta_a).
\end{align}
 Moreover, $\bu_a \in V$ and there exists a  dimensionless, absolute constant $C>0$ such that
\begin{align}
 \|\bu_\infty(t)\|_{\h^1} \le CR(r)\kappa_0\ \forall\ t \in \nR.
\end{align}
%  , and consider the trajectory $(\bu(t),\theta(t)):=S(t)(\buInit,\theta_0)$. 
%  Set $r=\|\thetaInit\|_{L^2}$ so that $(\bu(t),\theta(t))\in B_r$ for all $t\geq0$.
\comments{\begin{enumerate}[(i)]
\item 
   Suppose there is a solution $(\bu(t),\theta(t))\in B_r                          $ of \eqref{Bouss} and sequence of times $t_n$ satisfying $t_n\maps\infty$ as $n\maps \infty$, such that
\begin{align*}
 (\bu(t_n),\theta(t_n))\rightharpoonup (\bu_a,\theta_a) \text{ weakly in } H\times L^2
 \text{ as }n\maps\infty.
 \end{align*}

 \item Let $(\buInit,\theta_0)\in (V\cap \h^2)\times \h^2$.  Then there exists a point $(\bu_a,\theta_a)\in\cA$ such that $\theta_a$ has the same distribution function as $\theta_0$; that is, for all $\rho\in\nR$,
 begin{align}F_{\theta_a(\cdot,t)}(\rho) = F_{\theta_0}(\rho)\end{align}
 \end{enumerate}
 }

 \end{lemma}

%%%%%%%%%%%%%%%%%%%%%%%%%%%%%%%%%%%%%%%%%%%%%%%%%%
\begin{proof}
We will start by recalling certain {\it a priori} bounds; see \cite{Constantin_Foias_1988} and  \cite{Robinson_2001}. Let 
$(\bu(t),\theta(t)), t \ge 0$ be any trajectory starting in $B_r$ and observe that $(\bu(t),\theta(t)) \in B_r$ due to the semi-invariance of $B_r$. In particular, this means
\begin{gather}  \label{timeinv}
\|\bu(t)\|_{L^2} \le R(r) \ \mbox{and}\ \|\theta(t)\|_{\h^1} \le r\qquad (t \ge 0),
\end{gather}
\comments{Due to Proposition \ref{prop_theta_dist_conserved},
it follows that $\|\theta(t)\|_{L^2}$ as well as the distribution function $F_{\theta(t)}$ remains invariant for all $t\geq 0$.  
Moreover, 
using Poincar\'e and Young's inequalities and denoting $\kappa_0=\frac{2\pi}{L}$, 
the momentum equation \eqref{Bouss_mo} readily yields the time invariant $L^2$ bound on $\bu$, namely,
\begin{gather}  \label{timeinv}
\|\bu(t)\|_{L^2}^2 \le \|\buInit\|_{L^2}^2 + \frac{g^2\|\thetaInit\|_{L^2}^2}{\nu^2\kappa_0^4},
\end{gather}
}
Moreover, the momentum equation \eqref{Bouss_mo} readily yields the bound
\begin{gather} \label{energyineq}
\|\bu(t)\|_{L^2}^2 + \nu  \int_0^t \|\bu(s)\|_{\h^1}^2\, ds 
\leq
\|\buInit\|_{L^2}^2 + \frac{g^2\|\thetaInit\|_{L^2}^2}{\nu\kappa_0^2}T,
\ ( 0<t \le T),
\end{gather}
where the invariance of the $L^2$ norm of $\theta(\cdot)$ along a trajectory is also used.
Thus, there exists $t \in (0,T)$ such that 
\begin{gather}  \label{dblnormbd}
\|\bu(t)\|_{\h^1}^2 \le 2\left[ \frac{\|\buInit\|_{L^2}^2}{\nu T}+ \frac{g^2\|\thetaInit\|_{L^2}^2}{\nu^2\kappa_0^2}\right].
\end{gather}
Since the distribution function of $\theta$ remains invariant along a trajectory forward in time, by replacing the initial data if necessary,
we can (and henceforth will) without loss of generality assume that $\buInit \in V$ and moreover, setting $T= \frac{1}{\nu\kappa_0^2}$ in \eqref{dblnormbd}, we may also assume
\begin{gather}  \label{assumedblnormbd}
\|\buInit\|_{\h^1} \le 2R(r)\kappa_0,
\end{gather}
where $R(r)$ is as in Definition \eqref{def_B_r}.
With the assumption \eqref{assumedblnormbd}, noting that the periodic boundary condition and \eqref{Bouss_div} implies (in 2d) that $(B(\bu,\bu), A \bu)=0$, and consequently, following techniques in \cite{Constantin_Foias_1988} (Chapter XIII, pages 111-113), we additionally have the uniform bound
\begin{gather}  \label{timeinvdblbd}
\sup_{t \ge t_1}\|\bu(t)\|_{\h^1} \le C R(r)\kappa_0,
\end{gather}
where  $R(r)$ is as in Definition \ref{def_B_r}, $t_1=\frac{1}{\nu\kappa_0^2}$. 
 Additionally, from the first relation in \eqref{timeinvdblbd} and using again $(B(\bu,\bu),A\bu)=0$,  for any 
$\tau \ge 0$ and $T \ge 0$, we have
\begin{gather}  \label{vbound}
\|\bu(\tau+T)\|_{\h^1}^2 + 2\nu \int_\tau^{\tau+T} \|A \bu\|_{L^2}^2 
\leq 
C^2R(r)^2\kappa_0^2 + \frac{g^2\|\thetaInit\|_{L^2}^2}{\nu}T.
\end{gather}

Let now $(\bu_{0,n},\theta_{0,n}) \in B_r$ and the sequence of times $t_n \ra \infty$ be such that
\eqref{omegaconv} holds. 
Translating in time by $\tau_n\in[0,\frac{1}{\nu\kappa_0^2}]$ ,if necessary, we will assume 
without loss of generality
that $\bu_{0,n} \in V$, and that the bound in \eqref{assumedblnormbd} holds for all $\bu_{0,n}$.
Since $\bu_{0,n} \overset{\text{wk}}{\lra} \bu_a$, it immediately follows that 
\begin{align}
\bu_a \in V\ \mbox{and}\ 
\|\bu_a\|_{\h^1} \le 2R(r)\kappa_0.
\end{align}
\comments{We will first show that $\bu_a \in V$.
Let 
begin{align}
\bu^{(n)}:=\ \mbox{the}\, \bu\  \mbox{component of}\  S(t_n)(\bu_{0,n},\theta_{0,n}).
\end{align}
For sufficiently large $n$ (and consequently $t_n$), by the second bound in \eqref{timeinvdblbd}, $\|A\bu^{(n)}\|$ is bounded uniformly  in $n$. Consequently, 
by compactness of the embedding $D(A) \hookrightarrow V$, there exists a subsequence $\bu^{(n_j)}$ which converges in $V$. From \eqref{omegaconv}, $\bu^{(n)} \overset{\text{wk}}{\lra} \bu_a$ (thus, $\bu_a \in H$). By uniqueness of weak limits, the limit in $V$ of $\bu^{(n_j)}$ must be $\bu_a$  and thus 
$\bu_a \in V$.
We will now prove the remainder of the theorem.}  Set  
 \begin{align}\label{def_subseq_shift}
 ( \bu_n(t),\theta_n(t)):=S(t+t_n)(\bu_{0,n},\theta_{0,n}),\qquad t \in [-t_n, \infty).
\end{align}
Note that the trajectories $( \bu_n(t),\theta_n(t))$ are defined, for $n$ sufficiently large, 
on the intervals 
\begin{align}
 I_M:=\left(-\frac{M}{\kappa_0^2\nu},\frac{M}{\kappa_0^2\nu}\right), \qquad
M\in\nN.
\end{align}
By \eqref{timeinv} and \eqref{timeinvdblbd}, for $n$ sufficiently large, the pair
$(\bu_n,\theta_n)$ is bounded uniformly in $L^\infty(I_M,V)\times L^\infty(I_M,L^2)$,
and moreover due to \eqref{vbound}, $\bu_n$ is bounded uniformly in $L^2(I_M; {\mathbb H}^2)$. The above mentioned 
bounds are uniform with respect to $n$, although they may depend on $M$ in general.
From the functional equation \eqref{functmom}, the time independent $L^2$ bound \eqref{timeinv} on $\theta_n(t)$, the (uniform in $n$) bound on $\bu_n$  in 
$L^2(I_M; {\mathbb H}^2)$, and the Ladyzhenskaya inequality \eqref{ladyzhenskaya},
we readily obtain that $\frac{d}{dt} \bu_n $ is uniformly bounded in $L^2(I_M; H)$.
Consequently, $\{\bu_n\} $ is an equicontinuous family (in time) in $C(I_M; H)$. 
Moreover, from the Agmon inequality \eqref{agmon}, and the
uniform $L^2$ bound on $\bu_n$ and  \eqref{Bouss_den}, it also follows that $\frac{d}{dt} \theta_n$ is uniformly bounded in $L^4(I_M; {\mathbb H}^{-1})$ (in fact, one can show that the bound is uniform in
$L^p(I_M; {\mathbb H}^{-1})$ for any $p < \infty$). 
Thus, $\{\theta_n\}$ form an equicontinuous family in $C(I_M; {\mathbb H}^{-1})$. Now, by the
 Arzela-Ascoli Theorem and the Cantor diagonalization procedure, there exists a pair
$(\bu_\infty(\cdot),\theta_\infty(\cdot)) \in C((-\infty,\infty); H) \times  C((-\infty,\infty); {\mathbb H}^{-1})$
 and a subsequence (which we also denote by $(\bu_n,\theta_n)$) such that for all $M \in \nN$,
 \begin{gather}  \label{convergence}
(\bu_n(\cdot),\theta_n(\cdot))\overset{n \ra \infty}{\lra}(\bu_\infty(\cdot),\theta_\infty(\cdot))\  \mbox{in}\
 C(I_M; H) \times  C(I_M; {\mathbb H}^{-1}).
 \end{gather}
Note that since $\bu_n(0)$ is equal to the $\bu$ component of $S(t_n)(\bu_{0,n},\theta_{0,n})$, by  
the limits assured in \eqref{omegaconv} and  \eqref{convergence}, we have $\bu_\infty(0)=\bu_a$.
On the other hand, due to the hypothesis and \eqref{convergence}, 
\begin{align}
\theta(t_n)=\theta_n(0) \overset{\text{wk}}{\lra} \theta_a\ \mbox{and}\ 
\theta(t_n)=\theta_n(0) \overset{\text{in}\, \h^{-1}}{\lra} \theta_\infty(0)
\text{ as }n\ra\infty.
\end{align}
Thus, $\theta_\infty(0)=\theta_a$.

Next,  we will show that, by passing through a further subsequence if necessary, we can ensure that
   for all $\eta \in L^2$, the functions $\ip{\theta_\infty(\cdot)}{\eta} \in C((-\infty,\infty); \nR)$, and that 
\begin{gather}  \label{thetastrongweak}
\ip{\theta_n(\cdot)}{\eta} \ra \ip{\theta_\infty(\cdot)}{\eta}\ \mbox{ in}\  C(I_M; \nR),
\text{ as }n\ra\infty,
%\mbox{and} \ \theta_n \ra \theta_\infty\ \mbox{in}\ L^2(I_M; L^2)
\end{gather}
and also that $(\bu_\infty(t),\theta_\infty(t)) \in B_r\ \forall\ t \in \nR$, i.e., 
\begin{gather} \label{capsheafuniform}
\|\theta_\infty (t)\|_{L^2}^2 \le r^2\ \mbox{and}\ \|\bu_\infty(t)\|_{L^2}^2 \le 4 \frac{g^2r^2}{\nu^2\kappa_0^4}\ \forall\ t \in \nR.
\end{gather}
To see this, note that for any $\zeta \in {\mathbb H}^3$ (which, in 2D,implies that $\nabla \zeta \in L^\infty$), we have from \eqref{Bouss_den},
\begin{align}
& \quad
|\ip{\theta_n(t_2)-\theta_n(t_1)}{\zeta}| \le \int_{t_1}^{t_2} |\ip{\theta_n(s)\bu_n(s)}{ \nabla \zeta }|ds\\
 &\leq\notag
 \|\theta_n(\cdot)\|_{L^{\infty}(I_M;L^2)}\|\bu_n(\cdot)\|_{L^{\infty}(I_M;H)}\|\nabla \zeta\|_{L^\infty}|t_2-t_1|.
\end{align}
Thus, for all $\zeta \in {\mathbb H}^3$, the functions $\{(\theta_n(\cdot), \zeta) \}$ are equicontinuous
on $I_M$. Due to the uniform bound on $\|\theta_n(t)\|_{L^2}$, we have a pre-compact family. Let
$\{\zeta_j\}_{j=1}^\infty$ be a set of functions in ${\mathbb H}^3$ which is dense in $L^2$. 
By the Arzela-Ascoli Theorem and the Cantor diagonalization process, we may also assume that
\begin{align}
\ip{\theta_n(\cdot)}{ \zeta_j} \ra \ip{\theta_\infty(\cdot)}{ \zeta_j} \ \mbox{in}\ 
C(I_M; \nR)\text{ for all } M, j  \in  \nN.
\end{align}
To complete the proof, we show that $\ip{\theta_n(\cdot)}{ \zeta} \ra \ip{\theta_\infty(\cdot)}{ \zeta}$
uniformly on $I_M$ for any $\zeta \in L^2$. Let $\epsilon>0$ and choose $j$ such that
$|\zeta-\zeta_j| \le \frac{\epsilon}{2\|\thetaInit\|_{L^2}}$. Then,
\begin{align}
& |\ip{\theta_n(t)-\theta_m(t)}{\zeta}| \le |\ip{\theta_n(t)-\theta_m(t)}{\zeta_j}|
+ \|\theta_n(t)-\theta_m(t)\|_{L^2}\|\zeta_j-\zeta\|_{L^2}
\\\notag
& \le  |\ip{\theta_n(t)-\theta_m(t)}{\zeta_j}|+2\|\thetaInit\|_{L^2}\frac{\epsilon}{2\|\thetaInit\|_{L^2}}.
\end{align}
This shows that $\set{\ip{\theta_n(\cdot)}{\zeta}}$ is uniformly Cauchy in $C(I_M)$. Since $\theta_n(\cdot)$ converges to $\theta_\infty(\cdot)$ in 
$C(I_M;{\mathbb H}^{-1})$, the uniform limit of $\ip{\theta_n(t)}{\zeta}$ is $\ip{\theta_\infty(t)}{\zeta}$.
Due to weak convergence of $\theta_n(t)$ and strong convergence of $\bu_n(t)$, the uniform bounds 
\eqref{capsheafuniform} readily follow. This finishes the proof of the claim.

Next,  we will show that $(\bu_\infty,\theta_\infty)$ is a global solution of \eqref{Bouss}, i.e., we need to show that  $(\bu_\infty,\theta_\infty)$ satisfies \eqref{Bouss_wk} for appropriate test functions $\Phi$ and $\phi$ (see Definition \ref{def:Bouss_wk}).
Note that $(\bu_n,\theta_n)$ satisfy the weak formulation \eqref{Bouss_wk}. Passing to the limits in the linear terms, due to \eqref{convergence}, we find
\begin{subequations}
\begin{align}
\int_0^T\ip{\bu_n(s)}{ \Phi'(s)}\,ds
  &\maps
  \int_0^T\ip{\bu_\infty(s)}{ \Phi(s)}\,ds
   ,\\
\nu\int_0^T\iip{\bu_n(s)}{ \Phi(s)}\,ds
&\maps
\nu\int_0^T\iip{\bu_\infty(s)}{\Phi(s)}\,ds
,\\
\int_0^T\ip{\theta_n(s)\bg}{\Phi(s)}\,ds
&\maps
\int_0^T\ip{\theta_\infty(s)\bg}{\Phi(s)}\,ds
,\\
\int_0^T\ip{\theta_n(s)}{\phi'(s)}\,ds
&\maps
\int_0^T\ip{\theta_\infty(s)}{\phi'(s)}\,ds.
\end{align}
\end{subequations}

It remains to show the convergence of the remaining non-linear terms.  Let
\begin{subequations}
\begin{align}
   I(n)&:=\int_0^T\ip{\bu_n\otimes\bu_n-\bu_\infty\otimes\bu_\infty}{\nabla \Phi(s) }\,ds,
   \\  
J(n)&:=\int_0^T\ip{\bu_n(s)\theta_n(s)-\bu_\infty(s)\theta_\infty(s)}{\nabla\phi(s)}\,ds.
\end{align}
\end{subequations}
The convergence $I(n)\maps 0$ as $n\maps\infty$ follows from the convergence of the $\bu_n$ component guaranteed in \eqref{convergence}.
%is similar to standard arguments in
%the theory
%of the Navier-Stokes equations, thanks to \eqref{st_u_L2V} 
%(see, e.g., \cite{Temam_2001_Th_Num, Constantin_Foias_1988}), and thus we omit it.  
To show $J(n)\maps0$ as $n\maps\infty$, we write
$J(n)=J_1(n)+J_2(n)$, the definitions of which are given below. 
We have
\begin{align}
   J_1(n)
   &:=
   \int_0^T\ip{( \bu_n(s)-\bu_\infty(s) ) \theta_n(s)}{\nabla \phi(s)}\,ds
     \maps0
\end{align}
as $n\maps\infty$, since $\bu_n\maps\bu$  in $C(I_M,H)$ and
$\theta_n$ is uniformly bounded in $L^\infty(I_M,H)$.  For $J_2$, we have
\begin{align}
   J_2(n)&:= \int_0^T\ip{\bu_\infty(s)(\theta_n(s)-\theta_\infty(s))}{\nabla
\phi(s)}\,ds
     \maps0,
\end{align}
due to the fact that $\phi$ is a smooth test function for each fixed $M$, $\sup_{s \in I_M}\|\bu_\infty(s)\|_{\h^1} < \infty$  and $\theta_n (\cdot) \ra \theta_\infty(\cdot)$ in $C(I_M; \h^{-1})$.

Since $(\bu_\infty(t),\theta_\infty(t))$ is a global trajectory with  
\begin{align}
(\bu_\infty(t),\theta_\infty(t)) \in B_r \ \mbox{ for all}\  t \in \nR,
\end{align}
 it readily follows from \eqref{timeinvdblbd} that $\|\bu_\infty(t)\|_{\h^1} \le CR(r)\kappa_0$ for all $t \in \nR$.
\end{proof}
We will also need the next proposition in order to proceed.
\begin{proposition}  \label{wkcontinuity}
For each fixed $t$ and  a sequence $(\bu_{0,n},\theta_{0,n}) \overset{\text{wk}}{\lra} (\buInit,\theta_0)$, we have 
\begin{align}
S(t)(\bu_{0,n},\theta_{0,n}) \overset{\text{wk}}{\lra} S(t)(\buInit,\theta_0). 
\end{align}
In particular, for each fixed $t$, the map $S(t):X \lra H \times L^2$ is weakly continuous, where $X$ is a bounded subset of 
$H \times L^2$.
\end{proposition}
\begin{proof}
Fix $T >0$ and let $(\bu_{0,n},\theta_{0,n})$ be the sequence in the hypothesis of the proposition.  Due to the weak convergence, the sequence is uniformly bounded in the phase space.
Let $(\bu_n(t),\theta_n(t)):=S(t)(\bu_{0,n},\theta_{0,n}), 0 \le t \le T$. By arguments similar to the proof of Lemma \ref{lem_infty_solutions}, and by the uniform (in $n$) bounds provided in Theorem \ref{thm_Danchin_Paicu}, $(\bu_n(\cdot),\theta_n(\cdot))$ converges (in $C[\epsilon,T];H \times \h^{-1})$
for all $\epsilon >0$) to a solution $(\bu(\cdot),\theta(\cdot))$ with initial data $(\buInit,\theta_0)$. By uniqueness of solutions, 
$(\bu(t),\theta(t))=S(t)(\buInit,\theta_0), 0 \le t \le T$. Since
 $(\bu_n(T),\theta_n(T)) \overset{\text{wk}}{\lra} (\bu(T),\theta(T))$ (see proof of Lemma \ref{lem_infty_solutions}), the proof is complete.
 The weak continuity of $S(T)$ on bounded subsets now follows from noting that the weak topology is metrizable on bounded sets.
\end{proof}
 
 % =====================================================================
\subsection{Proof of Theorem \ref{main_att_thm}}
% =====================================================================
For the remainder of this section, we focus on proving Theorem \ref{main_att_thm}. In particular, we will examine to what extent  $\cA$ shares the properties \eqref{A_attracting}, \eqref{A_bounded}, and \eqref{A_invariant} with $\cA_{\text{NS}}$, the attractor for the 2D Navier-Stokes equations.

% ==============================================================
\comments{\begin{theorem}\label{thm_A_eq_U_Ar}
The relation $\cA = \bigcup\limits_{r\geq0} \cA_r$ holds, where $\cA_r$ is as defined in \eqref{ardef}.
\end{theorem}
}
%The proof below is complicated by the fact that we have only proved backward uniqueness in $(V\cup \h^3)\times \h^3$ (Theorem \ref{thm_backward_uniqueness}), but it is not known if it holds in $H\times L^2$.  
%Otherwise, much of the proof uses standard techniques, but we include them here since our definition of the attractor is in some sense unconventional (in particular, one must be careful, since the usual equivalent definitions of attractor are not equivalent in our case).

\noindent
{\sc Proof of Theorem \ref{main_att_thm}} (i)
\begin{proof}
We first show the inclusion in the ``$\supset$'' direction. Let $(\bu_\infty,\theta_\infty)\in\cA_r$ for some $r\geq 0$.  Then, since $\cA_r=\omega(B_r)$, due to \eqref{omegaset} and Lemma \ref{lem_infty_solutions}, it immediately follows that
$(\bu_{\infty},\theta_{\infty}) \in \cA$.

Next, we show the inclusion in the ``$\subset$'' direction.
Choose any $(\buInit,\thetaInit)\in\cA$, and suppose that $(\buInit,\thetaInit)\not\in\cA_s$ for all $s\geq0$.  By the definition of $\cA$, there exists a trajectory $(\bu(t),\theta(t))\in \cA$ such that for some $t_0$, $(\bu(t_0),\theta(t_0))=(\buInit,\thetaInit)$.  Furthermore, $(\bu(t),\theta(t))$ is bounded in $H\times L^2$ uniformly for all $t\in\nR$, so there exists an $r>0$ such that $(\bu(t),\theta(t))\in B_r$ for all $t\in\nR$.  Since $(\bu(t),\theta(t))$ is defined for $t\in\nR$, we must have $(\buInit,\thetaInit)=S(\tau)(\bu(t_0-\tau),\theta(t_0-\tau))\in B_r$ for every $\tau\in\nR$.  But, $(\bu(t_0-\tau),\theta(t_0-\tau))\in B_r$, so $(\buInit,\thetaInit)\in \omega(B_r)=\cA_r$.  
\end{proof}

\noindent
{\sc Proof of Theorem  \ref{main_att_thm}} (ii)
\comments{
\begin{theorem}\label{prop_steady_states}
$\cA$ contains all the steady states of \eqref{Bouss}.
\end{theorem}
}
\begin{proof} Since any steady state is bounded for all times $t\in\nR$, the proposition follows immediately from the definition of $\cA$.
\end{proof}
% ==============================================================

\noindent
{\sc Proof of Theorem \ref{main_att_thm}} (iii)
\comments{
\begin{theorem}\label{thm_no_weak_interior}
$\cA$ has empty interior in the strong (and therefore weak) topology of $H \times L^2$.
\end{theorem}
}
\begin{proof}
If $(\bu_a,\theta_a) \in \cA$, then by Lemma \ref{lem_infty_solutions}, $\bu_a \in V$. Any open ball in $H \times L^2$ (in the norm topology) must contain points whose first component is not in $V$. Thus, there does not exist an open ball contained in $\cA$.
\end{proof}

\noindent
{\sc Proof of  Theorem \ref{main_att_thm}} (iv)
\comments{
\begin{theorem}\label{prop_special_solutions}
 The weak sigma-attractor $\cA$ corresponding to the system \eqref{Bouss} is  a  non-empty, proper subset  of the phase space $X:=H\times L^2(\Omega)$ which moreover contains infinite-dimensional subspaces of the phase space. 
\end{theorem}
}
\begin{proof}
Theorem \ref{main_att_thm} (iii)  immediately implies that $\cA \neq X$. To show it is nonempty,
consider a solution of the form $(u^H_1,u^H_2,\theta^H,p^H)$, defined in Subsection \ref{steadystates}.  Let us also impose that $\partial_t u^H_2\equiv0$, so that $u_2^H=u_2$, where $u_2=u_2(x_1)$ is the unique, mean-free and periodic solution of the equation
\begin{equation}
 \nu\frac{d^2}{dx_1^2}u_2(x_1)=-g\theta^H(x_1).
\end{equation}
Then $(u^H_1,u^H_2,\theta^H,p^H)$ is a steady state solution of \eqref{Bouss}, so $(u^H_1,u^H_2,\theta^H)\in\cA$. Since $\theta^H\in L^2$ can be chosen arbitrarily, $\cA$ contains an infinite dimensional subspace.
\end{proof}

%%======================================================================

\noindent
{\sc Proof of Theorem \ref{main_att_thm}} (v)
\comments{
\begin{theorem}\label{thm_Ar_weakly_compact}
For every $r\geq0$, $\cA_r$ is weakly compact in $H\times L^2$.
\end{theorem}
}
\begin{proof}
Since $\cA_{r_1} \subset \cA_{r_2}$ if $r_1 < r_2$, it follows that $\cA = \bigcup\limits_{r\geq0,r \in \nQ} \cA_r$.
To complete the proof, 
it is enough to show that for every $r\geq0$, $\cA_r$ is weakly compact in $H\times L^2$.
By the Banach-Alaoglu Theorem, $B_r$ is weakly compact in $H\times L^2$.  Since $S(t)$ is weakly continuous, $S(t)B_r$ is weakly compact.  Thus, $A_r:=\omega(B_r)$ is the intersection of weakly compact sets, so it is weakly compact.
\end{proof}

\noindent
{\sc Proof of Theorem \ref{main_att_thm}} (vi)
\begin{proof}
Suppose $X\subset H\times L^2$ is bounded.  Then there exists $r>0$ such that 
$X \subset B_r$, where $B_r$ is the absorbing set as defined in Definition \ref{def_B_r}.
Thus, due to semi-invariance of $B_r$ under $S(t)$, we have $S(t)X \subset B_r$ for all $t>0$. 
The proof now proceeds by contradiction. Assume that the conclusion of the theorem is false.
Then there exists a sequence $(\bu_n,\theta_n) \in X$, $t_n \ra \infty$ and $\epsilon>0$ such that
\begin{gather}  \label{positivedist}
d(S(t_n)(\bu_n,\theta_n), \by) > \epsilon\  \mbox{for all}\  \by \in \omega(B_r)=\cA_r \subset \cA. 
\end{gather}
However, since 
$S(t_n)(\bu_n,\theta_n) \in B_r$, by the weak compactness of $B_r$, there exists $\by_0 \in B_r$
and a subsequence $n_j$ such that $S(t_{n_j})(\bu_{n_j},\theta_{n_j}) \overset{\text{wk}} {\lra}\by_0$, or equivalently, $d(S(t_{n_j})(\bu_{n_j},\theta_{n_j}),\by_0) \ra 0$. By the definition of $\omega(B_r)$, the point $\by_0 \in \omega(B_r)$ and this contradicts \eqref{positivedist}.
\end{proof}
% ==============================================================

% ==============================================================
\comments{
\begin{theorem}\label{thm_weakly_connected}
For each $r\geq0$, $\cA_r$ and $\cA$ are weakly connected.
\end{theorem}
}

\noindent
{\sc Proof of Theorem \ref{main_att_thm}} (vii)
\begin{proof}
Suppose $\cA_r$ is not weakly connected.  Then there exist open sets $O_1$ and $O_2$ in the weak topology of $H\times L^2$ such that $\cA_r\subset O_1\cup O_2$, $O_1\cap O_2=\varnothing$, $\cA_r\cap O_1\neq\varnothing$, and $\cA_r\cap O_2\neq\varnothing$.   Recall that $\cA_r:=\omega(B_r)$.  Since $B_r$ is weakly connected and $S(t)$ is weakly continuous, $S(t)B_r$ is also weakly connected.  Thus, for each $n\in\nN$, there exists $(\bu_n,\theta_n)\in (S(n)B_r)\setminus (O_1 \cup O_2)$.  

Since $B_{r}$ is invariant under $S(t)$, $(\bu_n,\theta_n)$ is bounded.  
Thus, by the Banach-Alaoglu Theorem, there exists a point $(\bu,\theta)\in H\times L^2$ and a subsequence $(\bu_{n_i},\theta_{n_i})\rightharpoonup(\bu,\theta)$ in the weak topology of $H\times L^2$.  Since $O_1\cup O_2$ is weakly open, $(\bu,\theta)\not\in O_1\cup O_2$.  But, by Theorem \ref{main_att_thm} (vi), $\cA_r$  attracts all bounded sets in the weak topology, so that the weak limit of $(\bu_n,\theta_n)$ must lie in $A_r\subset O_1\cup O_2$; a contradiction.  Therefore, $\cA_r$ is connected.

To show that $\cA$ is weakly connected, note that by Theorem \ref{main_att_thm}(i), $\cA=\cup_{r\geq0}\cA_r$.  Furthermore, each $\cA_r$ is connected and contains the zero element.  Therefore, if $O_1$ and $O_2$ are weakly open sets that separate $\cA$, all $\cA_r$ must be contained in the one containing the zero element, so $\cA$ must be contained entirely in either $O_1$ or $O_2$, so that $\cA$ is connected.
\end{proof}
% ==============================================================
\comments{
\begin{theorem}\label{thm_invariance}
The global attractor $\cA$ is invariant (i.e. $S(t)\cA =\cA$) for all $t \ge 0$. 
Moreover, it is the minimal set that attracts all bounded sets.
\end{theorem}
}

\noindent
{\sc Proof of Theorem \ref{main_att_thm}} (viii)
\begin{proof}
Let $(\buInit,\theta_0) \in \cA$. Then there exists a bounded, global trajectory $(\bu(t),\theta(t)), t \in \nR$ passing through $(\buInit,\theta_0)$. It is clear from the definition of $\cA$ that all points on this trajectory also lie on $\cA$.  In fact, since this trajectory is bounded, it lies entirely in some $B_r$ and consequently, in $\omega(B_r)=\cA_r$.
We may assume, by translating time if necessary,  that 
$(\bu(0),\theta(0))=(\buInit,\theta_0)$.
Then, by (forward in time) uniqueness of solutions, 
\begin{align}
S(t)(\bu(-t),\theta(-t)) = (\bu(0),\theta(0))=(\buInit,\theta_0).
\end{align}
As remarked above, $(\bu(-t),\theta(-t)) \in \cA$. Thus, $\cA \subset S(t)\cA$ for all $t \ge 0$. 
For the reverse inclusion, simply note that the (forward in time) uniqueness of solutions guarantees that
$(\bu(t),\theta(t))=S(t)(\buInit,\theta_0)$ for all $t \ge 0$,  where $(\bu(t),\theta(t)), t \in \nR$ is the above mentioned bounded global trajectory. Thus $S(t)\cA=\cA$.

The fact that $\cA$ attracts all bounded sets was proven in Theorem \ref{main_att_thm} (vi). 
We will now prove that it is minimal. Let $\cA'$ be a set  which has the property that it attracts all bounded sets and let $(\buInit,\theta_0) \in \cA$. Then there exists a bounded, global trajectory passing through it. In view of the previous paragraph, we see that the set in the phase space
\begin{align}
{\mathcal G}=\{(\bu(t),\theta(t)): t \in \nR\} \subset H \times L^2,
\end{align}
is invariant under the semigroup $S(t)$ for all $t \ge 0$, i.e., $S(t){\mathcal G} ={\mathcal G}$ for all $t \ge 0$. On the other hand, ${\mathcal G}$ is a bounded set in $H \times L^2$. Thus, by  the assumption that $\cA'$ attracts all bounded sets,  
\begin{align}
\limsup_{t \ra \infty} \text{dist}(S(t){\mathcal G}, \cA')=0. 
\end{align}
Thus, ${\mathcal G} \subset \cA'$, and consequently, $(\buInit,\theta_0) \in \cA'$.
\end{proof}

\noindent
{\sc Proof of Theorem \ref{main_att_thm}}(ix)
\begin{proof}
Note first that since $(\bu(t),\theta(t)), t \ge 0$ is bounded in $H \times L^2$, there exists $(\bu_a,\theta_a)$ and a sequence $t_n$ such that $(\bu(t_n),\theta(t_n)) \overset{\text{wk}}{\lra} (\bu_a,\theta_a)$. 
By Lemma \ref{lem_infty_solutions}, we immediately infer that $(\bu_a,\theta_a) \in \cA$. Furthermore, 
the global trajectory $(\bu_\infty(\cdot),\theta_\infty(\cdot))$ constructed in Lemma \ref{lem_infty_solutions} lies in $\cA$ and
\begin{align}
(\bu_n(t),\theta_n(t)):=(\bu(t+t_n),\theta(t+t_n), t \in [-t_n,\infty),
\end{align}
 converges to  $(\bu_\infty(\cdot),\theta_\infty(\cdot))$ in
 $C([-M,M];H) \times C([-M,M];{\mathbb H}^{-1})$ for all $M>0$. This concludes the proof.
\end{proof}

% ==============================================================
% ==============================================================

% \todo[inline]{From Ciprian: ``Put as a remark something like this: We study the behavior in time of a set which is bounded.  Then it is in some other bounded set after awhile.  This set is metrizable.''}

%\todo[inline]{Is the (Hausdorff or fractal) dimension $\cA_r$ finite? How does it behave with $r$?}

\comments{
\subsection{Further Analysis of the Attractor}.
Suppose for some $r>0$ and some $(\bu,\theta)\in\cA_r$ that there exists a $t_0\in\nR$ such that $\bu(t_0)=\vect{0}$.  Then, from \eqref{Bouss}, we have
\begin{subequations}
 \begin{align}
 \partial_t\bu\big|_{t=t_0} &= P_\sigma(\theta(t_0) \bg),
 \\
  \partial_t\theta\big|_{t=t_0} &= 0,
\end{align}
\end{subequations}
Differentiating \eqref{Bouss_mo} with respect to $t$ yields
\begin{align}
\bu_{tt} + B(\bu_t,\bu) + B(\bu,\bu_t) &= \nu\triangle\bu_t+ \theta_t \bg
\end{align}
so that
\begin{align*}
 \bu_{tt}(t_0) &= \nu\triangle\bu_t(t_0)
\end{align*}
}

 % =====================================================================
 \section{The Presence of 2D Turbulence}  \label{sec:2dturb}
% =====================================================================
The Batchelor-Kraichnan-Leith theory \cite{Bat, Kr} of 2D turbulence (inspired by that of Kolmogorov in 3D \cite{K41_3,K41_1, K41_2}) asserts that, on average, the behavior of eddies in turbulent flows is determined by their length scales. In a relatively large range of scales $[\kappa_{*},\kappa^*]$ (called the {\it inertial range}) 
viscous effects are negligible and enstrophy is transferred at a nearly constant rate $\eta$ from one length scale to the next smaller one (termed the {\it enstrophy cascade}). The {\it dissipation range} consists of the very small length scales where the viscosity annihilates the enstrophy. Heuristic arguments by Batchelor and Krachnan \cite{Bat, Kr} place the dissipation range beyond a wave number $\kappa_\eta = \left(\frac{\eta}{\nu^3}\right)^{1/6}$ where $\eta$ is the average rate of enstrophy dissipation per unit mass. The main tenets of this empirical theory can thus be summarized as (see, e.g., \cite{Balci_Foias_Jolly_2010,Foias_Manley_Rosa_Temam_2001,foias5,foias4}),
\begin{itemize}
\item[(i)] a significant amount of enstrophy is in the inertial range $[\kappa_*,\kappa^*]$;
\item[(ii)] this range is wide, i.e. $\kappa_* << \kappa^* \sim \kappa_\eta$;
\item[(iii)] the direct cascade of enstrophy (to smaller scales) holds over this range;
\item[(iv)] the power law $e_{\kappa,2\kappa} \sim \frac{\eta^{2/3}}{\kappa^2}$ holds for the amount of energy $e_{\kappa,2\kappa}$  contained in the length scales $\kappa $ to $2\kappa$ for $\kappa \in [\kappa_*,\kappa^*]$.
\end{itemize}

Rigorous justification for parts of the theory has been obtained in the series of works 
\cite{foias2, foias3,Foias_Manley_Rosa_Temam_2001,foias5,foias4} among others. In \cite{foias2,foias1}, it is shown that many of the ubiquitous averages in empirical turbulence theory, upon the application of which patterns are observable, can be taken to be finite time averages, albeit on sufficiently long periods of time.  Furthermore in \cite{Balci_Foias_Jolly_2010} extension of the above mentioned works has been obtained for forcing in all scales. Thus, in view of the results in \cite{Balci_Foias_Jolly_2010, foias2, foias1}, the universal features of turbulence already hold for finite time averages of the form 
\begin{align}
\langle \Phi \rangle := \frac{1}{t_2-t_1}\int_{t_1}^{t_2} \Phi(\bu(s))ds,
\end{align}
where $\Phi$'s are the relevant physical functionals   
on the phase space $H $ and $t_2> t_1>0$ satisfies 
\[
\max\left\{ t_1, t_2 -t_1\right\}>>\frac{G}{\nu\kappa_0^2}.
\]
Here $G$ is the Grashoff number which is a non-dimensionalized version of the $L^2$ norm of the  driving force in the Navier-Stokes equations.

In principle, the results in \cite{Balci_Foias_Jolly_2010} apply in our setting to the velocity equation written in the functional form in \eqref{functmom}. However,  for the inertial range to be sufficiently large for turbulent patterns to emerge, it is necessary (though not sufficient; see Remark \ref{rmk:nonturbulent} below) for the magnitude of the driving force, as measured by the {\it Grashof number}, to be large. Indeed, it is well-known that the 2D Navier-Stokes equations converge to a steady state if the force is time independent and the Grashof number is sufficiently small  \cite{Temam_1997_IDDS, Temam_1995_Fun_Anal, Temam_2001_Th_Num}. Accordingly, let us define the dimensionless (time-dependent) number $G_\sigma$ by
\begin{align}
 G_\sigma\equiv G_\sigma(t) := 
 \frac{\|P_\sigma(\bg\theta)(t)\|_{L^2}}{\nu^2\kappa_0^2}
=
 \frac{g}{\nu^2\kappa_0^2}
\pnt{\sqrt{\|\thetaInit\|_{L^2}^2-\|R_2\theta(t)\|_{L^2}^2}}.
\end{align}
The $\sigma$ here is used to denote the fact that the norm of the solenoidal projection of the force is taken, rather than the norm of the force itself.  Note that, as in \eqref{def_L2_Grashof},one may also consider the somewhat simpler (though potentially larger) time-invariant, dimensionless number
\begin{align}
 G := 
 \frac{g\|\thetaInit\|_{L^2}}{\nu^2\kappa_0^2},
\end{align}
which obviously satisfies  $0 \le G_{\sigma}(t) \le G$. The {\it effective Grashof number} governing the dynamics of \eqref{functmom} is defined to be 
\begin{gather}  \label{effectivegrashoff}
G^*_\sigma = \limsup_{t \ra \infty}  G_\sigma(t).
\end{gather}

 Thus, the complexity of the flow, at least in regard to the statistical features, is expected to be determined by 
$G^*_\sigma $. This is borne out by the following result which is  analogous to the 2D Navier-Stokes equations.   Before stating this result, let us observe that since  $\theta \in L^2$, the distributional derivatives $\partial_{x_i}\theta, i=1,2$ belong to $\h^{-1}$ and moreover,
\begin{align}
\|\partial_{x_i}\theta(t)\|_{\h^{-1}} \lesssim \|\theta(t)\|_{L^2} = \|\theta_0\|_{L^2}, i=1,2.
\end{align}
\begin{theorem} \label{thm:steadystate}
Assume that $G_\sigma^* < \epsilon$ where $\epsilon >0$ is sufficiently small and that 
$P_\sigma\bg\theta (t)$ converges weakly to $\bbf\ \in\ H$ as $t \ra \infty$. Then $\bu(t)$ approaches a steady state of the 2D Navier-Stokes equations with time-independent driving force $\bbf$. Consequently,  if $\partial_{x_1}\theta (t)$ converges to zero in $\h^{-1}$ as $t \ra \infty$, or equivalently, if  $\|P_\sigma(\bg\theta)(t)\|_{L^2}$ converges to zero, then $\bu(t)$ converges to zero in $H$.
\end{theorem}
\begin{proof}
For notational simplicity, denote $\mathbf{F}(t) = P_\sigma\bg\theta (t) $ and observe that $\mathbf{F}(t) \in H$ and $
\|\mathbf{F}(t)\|_{L^2} \le g \|\theta_0\|_{L^2}$. Moreover, since $\mathbf{F}(t) \ra \bbf$ weakly,  we also have 
\begin{align}
\|\bbf\|_{L^2} \le \limsup_{t \ra \infty} \|\mathbf{F}(t)\|_{L^2} \le \epsilon \nu^2\kappa_0^2.
\end{align}
Thus, by shifting time if necessary, we may assume without loss of generality, that
\begin{gather}
\max\{\|\bbf\|_{L^2}, \sup_{t\ge 0} \|\mathbf{F}(t)\|_{L^2}\} \le 2\epsilon \nu^2 \kappa_0^2.
\end{gather}
Furthermore, due to \eqref{timeinvdblbd} in the proof of Lemma \ref{lem_infty_solutions}, and by shifting time again if necessary, we also have the uniform bound
\begin{gather}  \label{uniform}
\sup_{t \ge 0} \|\bu(t)\|_{\h^1} \le C \epsilon \nu \kappa_0,
\end{gather}
where $C$ is a non-dimensional, absolute constant.
Recall also that on bounded subsets of $H$, the weak topology is metrizable and a bounded sequence
$\bv_n \ra \bv$ weakly if and only if $\|A^{-1/2}(\bv_n-\bv)\|_{L^2} \ra 0$.
 Thus, 
\begin{gather}  \label{fconverge}
\lim_{t \ra \infty} \|A^{-1/2}(\mathbf{F}(t) -\bbf)\|=0.
\end{gather}
Let $\bv(t)$ be the solution of the Navier-Stokes equation with initial data $\bu_0\equiv\mathbf{0}$, and the force given by $\bbf$. Denote $\bw=\bu - \bv$. Then $\bw$ solves
\begin{align}
\frac{d\bw}{dt} + \nu A\bw +B(\bu,\bw)+B(\bw,\bu)=\mathbf{F}(t)-\bbf,\quad \nabla \cdot \bw =0,\quad \bw(0)=\bu_0.
\end{align}
Taking inner product of the above equation with $\bw$ and by applying Young's inequality, we readily obtain
\begin{align}
\frac12 \frac{d}{dt} \|\bw\|_{L^2}^2 + \frac{\nu}{2} \|A^{\frac12}\bw\|_{L^2}^2
\le \frac{1}{2\nu}\|A^{-\frac12}(\mathbf{F}(t)-\bbf)\|^2_{L^2} + |\ip{B(\bw,\bu)}{\bw}|.
\end{align}
From \eqref{B424}, we get  $|\ip{B(\bw,\bu)}{\bw}| \le \|A^{\frac12}\bu\|_{L^2} \|\bw\|_{L^2}\|A^{\frac12}\bw\|_{L^2}$. 
Thus by Young's inequality, we have
\begin{align}
\frac12 \frac{d}{dt} \|\bw\|_{L^2}^2 + \frac{\nu}{4} \|A^{\frac12}\bw\|_{L^2}^2
& \le \frac{1}{2\nu}\|A^{-\frac12}(\mathbf{F}(t)-\bbf)\|^2_{L^2} + \frac{1}{\nu}\|A^{\frac12}\bu\|_{L^2}^2\|\bw\|_{L^2}^2
\\\notag
& \le \frac{1}{2\nu}\|A^{-\frac12}(\mathbf{F}(t)-\bbf)\|^2_{L^2} + C^2\epsilon^2 \nu\kappa_0^2\|\bw\|^2,
\end{align}
where the inequality in the second line is obtained using \eqref{uniform}. If $C^2\epsilon^2 \le \frac18$, by Poincar\'{e} inequality, we obtain
\begin{align}
 \frac{d}{dt} \|\bw\|_{L^2}^2 + \frac{\nu\kappa_0^2}{4} \|\bw\|_{L^2}^2 \le \frac{1}{\nu}\|A^{-\frac12}(\mathbf{F}(t)-\bbf)\|^2_{L^2}.
 \end{align}
 By Gronwall inequality and \eqref{fconverge}, it immediately follows that $\lim_{t \ra \infty} \|\bw (t)\|=0$. On the other hand, it is well-known   that
 if $\|\bbf\|_{L^2}\le \epsilon \nu^2\kappa_0^2$  and $\epsilon $ is sufficiently small, then the solution to the Navier-Stokes equations corresponding to the time-independent force   $\bbf$, converges to a fixed point
 \cite{Temam_1995_Fun_Anal, Temam_1997_IDDS, Temam_2001_Th_Num}. Thus, if $\epsilon$ is sufficiently small, $\bv$ converges to a fixed point $\bv_0 \in H$ which satisfies $\nu A\bv_0+ B(\bv_0,\bv_0)=\bbf$. Since $\bw(t)$ converges to zero in $H$, it follows that $\bu(t)$ converges to $\bv_0$ in $H$.
 
 For the second part of the theorem, observe that by \eqref{Riesz},
 \begin{align}
\|P_\sigma\bg\theta (t)\|_{L^2}= \|R_1\theta(t)\|_{L^2}= \|(\Delta)^{-1/2}\partial_{x_1}\theta (t)\|_{L^2}\sim 
 \|\partial_{x_1}\theta (t)\|_{\h^{-1}}.
 \end{align}
It  follows  that if   $\partial_{x_1}\theta (t)$ converges to zero in $\h^{-1}$ as $t \ra \infty$ if and only if 
$\|P_\sigma\bg\theta (t)\|_{L^2}$ converges to zero. Thus one can apply the first part of the theorem with $\epsilon =0$ to conclude that $\bu(t) \ra {\mathbf 0}$ in $H$ as $t \ra \infty$.
\end{proof}

\begin{remark}  \label{rmk:nonturbulent}
It is useful to examine Theorem \ref{thm:steadystate} in light of the examples presented in Subsection \ref{steadystates}. In the ``vertical solutions" presented in (i), due to \eqref{Riesz}, 
$P_\sigma\bg\theta (t) \equiv 0$ and in conformity to theorem \ref{thm:steadystate}, $\bu(t) \ra 0$. On the other hand, for the  ``horizontal solutions" presented in Subsection \ref{steadystates} (ii), due again to \eqref{Riesz}, $\|P_\sigma\bg\theta (t)\|_{L^2} \equiv \|\theta_0\|_{L^2}$. Consequently, $G_\sigma^*$ can be made arbitrarily large, yet the flow is not turbulent as it converges to a (laminar) steady state. Incidentally, in addition to the example provided by Marchioro \cite{Marchioro_1986_Ab_Turb1,Marchioro_1987_Ab_Turb2} with forcing in the first eigenmode of the Stokes operator, this provides another example of solution to the 2D Navier-Stokes equations with arbitrarily large Grashof number for which the solution converges to a steady state.

\end{remark}

\begin{remark}
If one considers equations \eqref{Bouss_mo} and \eqref{Bouss_div}, this
is  the usual 2D Navier-Stokes system with forcing $\bg\theta$, and so the
framework developed recently in \cite{Balci_Foias_Jolly_2010} for 2D turbulence
with forcing at all scales applies. Moreover, unlike the usual 2D Navier-Stokes equations, the forcing here is not {\it ad hoc}, but rather intrinsic to the system.
Noting that $\|\bg\theta\|_{L^2}=\|\bg\theta_0\|_{L^2}$ for all $t \ge 0$, it seemingly allows arbitrarily large Grashof number as well. However, due to the divergence free condition on the velocity field, the {\it effective driving force} is given by  $P_\sigma\bg\theta (t)$. The quantity 
$\|P_\sigma\bg\theta (t)\|_{L^2}$, although bounded above  by $\|\bg\theta_0\|_{L^2}$, can potentially approach zero for large time. Thus, if one can show that
 $G_{* \sigma}=\liminf_{t\ra \infty} \|P_\sigma\bg\theta (t)\|_{L^2}$ can be made arbitrarily large, then one can assert that  \eqref{Bouss_mo} and \eqref{Bouss_div}, or equivalently \eqref{functmom} contains the entirety of Grashof numbers involved in 2D turbulent dynamics. Whether or not arbitrarily large $G_{* \sigma}$ and a truly turbulent flow can be achieved, remains an open question at this point. As observed in Remark \ref{rmk:nonturbulent}, 
 one can indeed achieve arbitrarily large $G_{* \sigma}$; the corresponding flow however approaches  a steady state and is therefore not turbulent. The velocity profile is also laminar; thus it does not yield a turbulent Lagrangian dynamics either.

\end{remark}

 \section{The Energy-Enstrophy Plane}\label{sec_Energy_Enstrophy}
% =====================================================================
\noindent
In this section, we investigate the relationship between the energy and the
enstrophy in the conservative set.  Here, we follow many of the ideas of
\cite{Dascaliuc_Foias_Jolly_2005,Dascaliuc_Foias_Jolly_2007,
Dascaliuc_Foias_Jolly_2008,Dascaliuc_Foias_Jolly_2010} which investigated the
relationship between energy, enstrophy, and palenstrophy in the attractor of the
2D Navier-Stokes equations.  

For $(\bu,\theta)\in \cA$, let us define
\begin{align}
 \chi(t) := \frac{\|\bu(t)\|_{\h^1}^2}{\|\bu(t)\|_{L^2}}, 
 \quad\text{and}\quad
 \lambda(t) := \frac{\|\bu(t)\|_{\h^1}^2}{\|\bu(t)\|_{L^2}^2}.
\end{align}
% Taking the inner product of \eqref{Bouss_mo} with $\bu$ and with $A\bu$, we find
% \begin{subequations}
% \begin{align}
%  |\bu|\frac{d}{dt}|\bu| +\nu\|\bu\|_{\h^1}^2 &= (\theta\bg,\bu),
%  \\
%  \frac12\frac{d}{dt}\|\bu\|_{\h^1}^2 +\nu\|A\bu\|_{L^2}^2 &= (\theta\bg,A\bu).
% \end{align}
% \end{subequations}
After some computation, which was carried out in the context of the 2D Navier-Stokes equations in \cite{Balci_Foias_Jolly_2010}, it can be shown that
\begin{align}
 \frac{d\chi}{dt}
 \label{chi_equation}&=
 \frac{\nu}{2\|\bu\|_{L^2}}\pnt{\frac{\|\bg\theta\|_{L^2}^2}{\nu^2}-\chi^2-\left\|2A\bu-\frac{
\|\bu\|_{\h^1}^2}{\|\bu\|_{L^2}^2}\bu-\frac{\bg\theta}{\nu}\right\|_{L^2}^2}
 \\\notag&=
 \frac{\nu}{2\|\bu\|_{L^2}}\pnt{\frac{g^2\|\thetaInit\|_{L^2}^2}{\nu^2}-\chi^2-\left\|\bw-\frac{
\bg\theta}{\nu}\right\|_{L^2}^2}
\end{align}
where $\bw := \pnt{2A-\frac{\|\bu\|_{\h^1}^2}{\|\bu\|_{L^2}^2}}\bu$.  Let us consider the quantities 
$\|\bu\|_{L^2}$ and $\|\bu\|_{\h^1}$ as variables.  Level sets of $\chi$ correspond to curves along which $\|\bu\|_{L^2}$ and $\|\bu\|_{\h^1}$ are parabolically related.  Due to \eqref{chi_equation}, if $\chi(t)>g\|\thetaInit\|_{L^2}/\nu$, then $\chi$ must decrease.  Furthermore, due to the Poincar\'e inequality, we must always have $\lambda(t)\geq\lambda_1$.  Therefore, trajectories flow towards the set 
\begin{align}
 \Lambda := \set{(\bu,\theta): \|\bu\|_{\h^1}^2\leq \frac{g\|\thetaInit\|_{L^2}}{\nu}\|\bu\|_{L^2}, \quad
 \|\bu\|_{\h^1}^2\geq \lambda_1\|\bu\|_{L^2}^2 ,\quad
 \|\theta\|_{L^2}=\|\thetaInit\|_{L^2}}.
\end{align}
It follows that $\cA\subset\Lambda$.  
See Figure \ref{parabola_pic} for a depiction of the projection of $\Lambda$ onto the energy-enstrophy plane.  

\begin{figure}[htp]
\begin{tikzpicture}[line cap=round,line join=round,>=triangle 45,x=8.0cm,y=4.0cm]
\draw[->,color=black] (-0.2,0) -- (1.17,0);
\foreach \x in {-0.2,0.2,0.4,0.6,0.8,1}
\draw[shift={(\x,0)},color=black] (0pt,-2pt);
\draw[->,color=black] (0,-0.22) -- (0,1.2);
\foreach \y in {-0.2,0.2,0.4,0.6,0.8,1}
\draw[shift={(0,\y)},color=black] ;
\clip(-0.2,-0.22) rectangle (1.17,1.2);
\draw [fill=gray,fill opacity=0.5] plot [smooth,samples=100,domain=0:1](\x,{\x}) -- plot [smooth,samples=100,domain=0:1] (\x,{sqrt(\x)});
\draw (-0.15,1.19) node[anchor=north west] {$\frac{\|\mathbf{u}\|_{\h^1}^2}{\nu^2\kappa_0^2 G^2}$};
\draw (1.05,0.005) node[anchor=north west] {$\frac{\|\mathbf{u}\|_{L^2}^2}{\nu^2G^2}$};
\draw[smooth,samples=100,domain=0:1.17] plot(\x,{sqrt((\x))});
\draw [domain=-0.0:1.17] plot(\x,{(-0--1*\x)/1});
\draw (0.19,1.0) node[anchor=north west] {$\|\mathbf{u}\|_{\h^1}^2 = \frac{g\|\thetaInit\|_{L^2}}{\nu}\|\mathbf{u}\|_{L^2}$};
\draw (0.57,0.59) node[anchor=north west] {$\|\mathbf{u}\|_{\h^1} =\kappa_0\|\mathbf{u}\|_{L^2}^2$};
% \draw (0.18,0.40) node[anchor=north west] {$\Lambda$};
\end{tikzpicture}
\caption{The projection of the attractor lies inside the shaded region.}
\label{parabola_pic}
\end{figure}
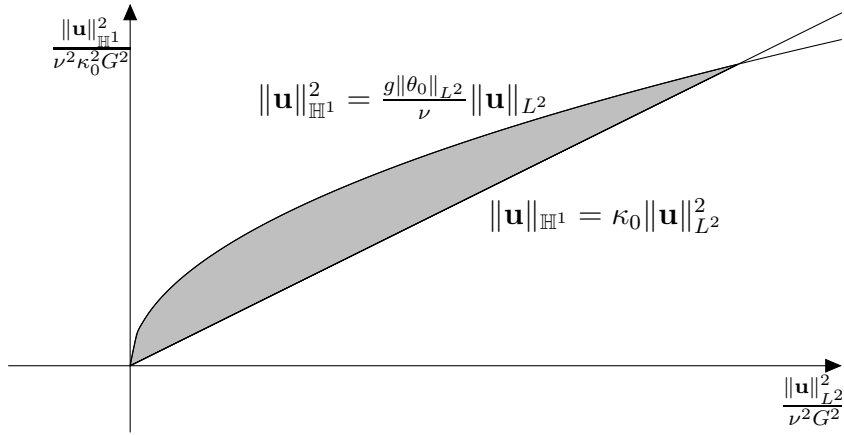

Let us consider the upper-boundary of this set; that is, the parabola $\|\mathbf{u}\|_{\h^1}^2 = \frac{g\|\thetaInit\|_{L^2}}{\nu}\|\mathbf{u}\|_{L^2}$. Suppose there is a point of $(\bu,
\theta)$ which lies on the parabola, and moreover that $\bu\equiv\buInit$ is independent of time (\textit{a priori}, $\theta$ may be balanced by a time-dependent pressure).  Then, since $\bu(t) = \buInit$, we have $\lambda(t) = \lambda(0)$.  On the parabola, equation \eqref{chi_equation} yields
\begin{align}\label{A_lambda_eqn}
0=
2A\bu-\frac{\|\bu\|_{\h^1}^2}{\|\bu\|_{L^2}^2}\bu-\frac{\bg\theta}{\nu}
=
2A\bu-\lambda(0)\bu-\frac{\bg\theta}{\nu}
\end{align}
so that, in fact, $\theta(t)=\theta_0$.
Taking the inner product with $\bu$ gives
\begin{align}\label{A_lambda_eqn_ip}
 2\|\bu\|_{\h^1}^2=\lambda(0)\|\bu\|_{L^2}^2+\frac{\ip{\bg\theta}{\bu}}{\nu}
%  \quad\Rightarrow\quad
%  \nu\|\bu\|_{\h^1}^2 = \ip{\bg\theta}{\bu}
  \quad\Rightarrow\quad
  \nu\lambda(0) 
  =
 \frac{\ip{\bg\theta}{\bu}}{\|\bu\|_{L^2}^2}.
\end{align}
% Taking the inner product of \eqref{Bouss_mo} with $\bu$ yields
% \begin{align}\label{energy_balance}
%  \frac12\frac{d}{dt}|\bu|^2  = \ip{\bg\theta}{\bu}-\nu\|\bu\|_{\h^1}^2=0.
% \end{align}
% Thus, $\|\bu(t)\|_{L^2}=\|\buInit\|_{L^2}$, which implies that 
% $\|\bu(t)\|_{\h^1}^2 = \frac{g\|\thetaInit\|_{L^2}\|\buInit\|_{L^2}}{\nu}=\|\buInit\|_{\h^1}^2$  on the parabola.
% Applying this to \eqref{energy_balance} shows that $\ip{\bg\theta}{\bu}$ is constant:
% \begin{align}
%  \ip{\bg\theta}{\bu} = \ip{\bg\thetaInit}{\buInit}
% \end{align}
Owing to \eqref{A_lambda_eqn_ip}, and the fact the $\bu$ is on the parabola, 
% and the fact that 
% $\lambda(t) 
% % = \frac{\chi(t)}{\|\bu(t)\|_{L^2}}
% % =\frac{\chi(0)}{\|\buInit\|_{L^2}}
% =\lambda(0)$
it follows that
\begin{align}
\frac{\ip{\bg\theta}{\bu}}{\nu\|\bu\|_{L^2}^2}
=
\lambda(0) 
=\frac{\|\bu\|_{\h^1}^2}{\|\bu\|_{L^2}^2}
=\frac{g\|\theta\|_{L^2}}{\nu\|\bu\|_{L^2}}.
\end{align}
Thus, 
\begin{align}
\ip{\bg\theta}{\bu}
=\|\bg\theta\|_{L^2}\|\bu\|_{L^2}.
\end{align}
By the Cauchy-Schwarz Theorem, we must have that $\bg\theta$ is a scalar multiple of $\bu$, say $\bg\theta=c\bu$ where $c=\frac{\ip{\bg\theta}{\bu}}{\|\bu\|_{L^2}^2}$.  
% The same argument shows that $\bg\theta=c\bu$ (with the same value for $c$).  
Applying this to \eqref{A_lambda_eqn}, we have
\begin{align}\label{A_lambda_eigenvalue}
2A\bu
=
\lambda(0)\bu+\frac{\bg\theta}{\nu}
=
\pnt{\lambda(0)+\frac{c}{\nu}}\bu
=\pnt{\frac{\|\bu\|_{\h^1}^2}{\|\bu\|_{L^2}^2}+\frac{\ip{\bg\theta}{\bu}}{\|\bu\|_{L^2}^2}}\bu
\end{align}
so that $\bu$ is an eigenfunction of $A$ with eigenvalue $\frac{1}{2}\pnt{\frac{\|\bu\|_{\h^1}^2}{\|\bu\|_{L^2}^2}+\frac{\ip{\bg\theta}{\bu}}{\|\bu\|_{L^2}^2}}$.

Going back to the equation $(0,g\theta)^T = \bg\theta=c\bu\equiv c(u_1,u_2)^T$, we observe that we must have $u_1\equiv0$.
% \footnote{For the sake of clarity, we have dropped the  ``0'' subscript, since $\bu$ and $\theta$ have now been shown to be independent of time.} 
The divergence-free condition then implies that $\partial_{x_2}u_2=0$, so that $\bu\cdot\nabla\bu \equiv0$.  Moreover, since $g\theta = c u_2$, we have $\partial_{x_2}\theta=0$, so that $\partial_t\theta = -\bu\cdot\theta =0$.  Thus, $\theta$ depends only on $x_1$, and again since $g\theta = c u_2$, so does $u_2$.  The entire system \eqref{Bouss} therefore reduces to the following relation:
\begin{align}
 \partial_{x_2}p-\nu\partial_{x_1}^2u_2 = g\theta.
\end{align}
Applying $\partial_{x_2}$ yields $\partial^2_{x_2}p=0$, so that the periodic boundary conditions  imply that $p$ is a constant.  The relation now becomes
\begin{align}
 -\nu\partial_{x_1}^2u_2 = g\theta= c u_2.
\end{align}
% Thus, for some $A,B\in \nR$, we have
% \begin{align*}
%  u_2 
%  = 
%  A\sin\pnt{\sqrt{\frac{c}{\nu}}x_1}+B\cos\pnt{\sqrt{\frac{c}{\nu}}x_1}
% \end{align*}
% Notice in particular that $u_2$ can be determined from no more than three parameters: $A$, $B$, and $c$ (of course, $\theta=cu_2/g$).  
The periodic boundary conditions further constrain $c$ to be of the form  
\begin{align}
 c = c_n := \nu\frac{n^2\pi^2}{L^2},\quad \text{ for some } n\in\nN.
\end{align}
Summarizing these observations, we have the following proposition, which essentially appeared in \cite{Dascaliuc_Foias_Jolly_2005} in the context of the Navier-Stokes equations.
\begin{proposition}
 Suppose $(\bu,\theta)$ is a smooth, steady state solution of \eqref{Bouss} which lies on the parabola $\|\mathbf{u}\|_{\h^1}^2 = \frac{g\|\thetaInit\|_{L^2}}{\nu}\|\mathbf{u}\|_{L^2}$
%  and that there exists a time $t_0$ such that it holds that 
%  \begin{align}
% \chi(t):=\|\bu(t)\|_{\h^1}^2/\|\bu(t)\|_{L^2}=g\|\thetaInit\|_{L^2}/\nu, \text{ for all } t\geq t_0                                                                                               \end{align}
Then $u_2$ and $\theta$ depend only on $x_1$, and %for all  $t\geq t_0$, 
 \begin{align}
  u_1\equiv0, \quad  u_2 = c_n\theta,\quad c_n:=\frac{gL^2}{\nu n^2\pi^2}\;\text{ for some } n\in\nN
 \end{align}
 and $u_2=u_2(x_1)$ is an eigenfunction of the operator $-\partial_{x_1}^2$ (respecting the periodic boundary conditions and mean-free condition).  
%  In particular, for all $t\geq t_0$,
%  \begin{align}
%   \chi(t) = \frac{n^2\pi^2}{L^2},\quad \text{ for some } n\in\nN.
%  \end{align}
\end{proposition}
\begin{remark}
 The above proposition implies that steady states on the upper boundary of $\Lambda$ are discrete points corresponding to geophysical steady states of stationary columnular flows.
 Somewhat more generally, it can be shown that if $(\bu,\theta)$ is a smooth flow for which $\theta$ depends upon $x_1$ alone, and $\theta_{x_1}\neq0 \text{ a.e.}$, then $(\bu,\theta)$ is one of the  horizontal flows described by \eqref{horizontal_heat_eqn} and \eqref{horizontal_ic}.
\end{remark}

%==============================================================
\section{Open Questions}\label{sec_open_questions}

In this paper, we have introduced a notion of an attractor for a semi-dissipative system, i.e.,  
a system which is a hybrid of  parabolic and  hyperbolic equations. While in some respects, it shares certain properties with the attractor of a dissipative system, there are several
 open questions concerning the attractor $\cA$. 
\begin{enumerate}[(i)]
\item We showed that the attractor $\cA$ is in some respects, a thin set, i.e., it has empty interior. 
This proof uses the fact that the velocity component of any point on $\cA$ is regular (belongs to $V$). On the other hand, we saw in 
subsection \ref{steadystates} that $\cA$ contains infinite dimensional subspaces of steady states. In particular, the projection of the attractor on the temperature component contains all functions of the type
$\theta=\theta(x_1)$ or $\theta=\theta(x_2)$ or  functions of the form $\theta(\bx)=h(\bk \cdot \bx)$, for an arbitrary function of one variable $h$. 
The question of how rich this set can be remains open.
In particular, if one takes the projection of $\cA$ onto the temperature component (say $P_\theta(\bu,\theta):=\theta$), the following questions remain open
\begin{enumerate}[(a)]
\item Does $P_\theta\cA$ have an empty interior in $L^2$?
 \item Is $P_\theta\cA$ a proper subset of $L^2$?
 \item Is $P_\theta\cA$ dense in $L^2$? 
 \end{enumerate}
 It is curious to note that in all the examples of steady states obtained here, the temperature $\theta$ is a function of only one variable.

\item We showed that although $\cA$ is unbounded and infinite dimensional, it is $\sigma$-compact in the weak topology. In particular, it can be written as a countable union of weakly compact omega limit sets $\cA_r=\omega(B_r)$ where $B_r$ are absorbing, invariant balls  as defined in \eqref{def_B_r}. The question is whether the attractor
has a pancake-like structure. In other words, is it true that $\cA = 
\bigcup_{r \geq 0}\omega(\tilde{B}_r)$, where 
\begin{gather}
 \tilde{B}_r:= \left\{(\buInit,\theta_0): \|\thetaInit\|_{L^2}=r, \|\buInit\|_{\h^1} \le R(r)=2\frac{gr}{\nu\kappa_0^2}\right\}?
\end{gather}
If this is the case, the $\bu$-component of $\omega(\tilde{B}_r)$ can be regarded as an attractor at {\it level } $r$ for the 2D NSE with time varying force $\theta$ whose $L^2$-norm remain fixed at $r$. The temperature component $\theta$ on the attractor is then a {\it rearranged} version of the initial temperature $\theta_0$. The problem with this picture is that an element of the omega limit set is obtained as a weak  limit of points $S(t)(\bu_{0,n},\theta_{0,n})$  with $(\bu_{0,n},\theta_{0,n}) \in \tilde{B}_r$. In the weak limit, the norms might decrease, i.e., it is possible that 
\begin{align}
\lim_{n \ra \infty} \|S(t)(\bu_{0,n},\theta_{0,n})\|_{L^2} >
\|\lim_{n \ra \infty} S(t)(\bu_{0,n},\theta_{0,n})\|_{L^2}. 
\end{align}
This might destroy the simplistic pancake structure described above. Whether or not this happens, and if it does, what  its implication is for the asymptotic dynamics, are questions that remain open.   Furthermore, although the whole attractor is infinite dimensional, is there any kind of finite dimensionality in its constituent pieces 
$\cA_r$?

\comments{
\item Given two solutions with the same velocity, say $(\bu,\theta_1)$ and $(\bu,\theta_2)$, the momentum equation immediately forces $\theta_1=\theta_2$.  However, the reverse question remains open; namely, given two solutions with the same temperature, say $(\bu_1,\theta)$ and $(\bu_2,\theta)$, does it follow that $\bu_1=\bu_2$?  This question is not determined in an obvious way from the momentum equation, due to the two-way nonlinear coupling between $\bu$ and $\theta$ in the transport equation.   If the answer is positive, then all the information of $\cA$ is contained in $P_\theta\cA$, which could have profound implications, e.g., for numerical simulations, due to the large reduction in the number of degrees of freedom.  On the other hand, if it is false, it would imply that the dynamics of the Boussinesq system are remarkably distinct from those of the Navier-Stokes equations.
}

\end{enumerate}

\section{Appendix}\label{sec_appendix}
% =====================================================================

\subsection{Absorbing Ball}\label{ssec_Grashof_Numbers}
For the sake of completeness, we include here a proof of the existence of an absorbing ball.
By taking the inner product of \eqref{Bouss_mo}, we obtain the energy estimate
\begin{align}
 \frac{1}{2}\frac{d}{dt}\|\bu\|_{L^2}^2 +\nu\|\bu\|_{\h^1}^2 
 &= \ip{P_\sigma(\theta\bg)}{\bu}
 \leq 
 \|P_\sigma(\bg\theta)\|_{L^2}\|\bu\|_{L^2}
 \\&\leq \notag
 \frac{\|P_\sigma(\bg\theta)\|_{L^2}^2}{2\nu\kappa_0^2} +
\frac{\nu\kappa_0^2}{2}\|\bu\|_{L^2}^2.
\end{align}
Using the Poincar\'e inequality \eqref{poincare}, \eqref{Riesz}, and
\eqref{Lp_conserved}, we obtain
\begin{align}\label{L2_after_Poincare}
 \frac{d}{dt}\|\bu\|_{L^2}^2 +\nu\kappa_0^2\|\bu\|_{L^2}^2 
 &
 \leq \frac{\|P_\sigma(\bg\theta)\|_{L^2}^2}{\nu\kappa_0^2}
  =
g^2\frac{\|\thetaInit\|_{L^2}^2-\|R_2\thetaInit\|_{L^2}^2}{\nu\kappa_0^2},
\end{align}
so that
\begin{align}
 \|\bu(t)\|_{L^2}^2
 \leq
 e^{-\nu\kappa_0^2t}\|\buInit\|_{L^2}^2
 +\frac{g^2}{\nu^2\kappa_0^4}\pnt{\|\thetaInit\|_{L^2}^2-\|R_2\thetaInit\|_{L^2}^2}(1-e^{
-\nu\kappa_0^2t}).
\end{align}
Thus,
\begin{align}\label{G_def}
\limsup_{t\maps\infty}\|\bu(t)\|_{L^2}^2
\leq 
\frac{g^2}{\nu^2\kappa_0^4}
\pnt{\|\thetaInit\|_{L^2}^2-\|R_2\theta(t)\|_{L^2}^2}\le\nu^2 G_\sigma^{*2}.
\end{align}
 
In particular, there exists a time $t_*=t_*(\|\buInit\|_{L^2})$ such that, for $t>t_*$, 
$\bu(t)\in B_{2\nu G_\sigma^*}$, the ball of radius $2\nu G_\sigma^*$ in $H$.  For example, $t_*$
can be taken as
\begin{align}
 t_*(\|\buInit\|_{L^2}) = \frac{1}{\nu\kappa_0^2}
 \max\set{1,\log\frac{\|\buInit\|_{L^2}^2}{3\nu^2G_\sigma^{*2}}}.
\end{align}

% -----------------------------------------------------------
\subsection{Mean Preservation}\label{ssec_Galilean_Invariance}
% -----------------------------------------------------------
We show that the mean is preserved by the flow.
% , and that we have hyper Galilean invariance.  
Since the equation is globally well-posed and possesses
higher-order regularity, we may assume our solutions are smooth enough for the
operations below to be justified rigorously.  We give the argument only in two
dimensions, but it extends without difficulty to the higher dimensional case.

%  \subsubsection{Mean Preservation}
 Integrating \eqref{Bouss_den} over $\Omega$, we find 
\begin{align}
\frac{d}{dt}\int_{\Omega}\theta\,dx =- \int_{\Omega}\nabla\cdot(\bu\theta)\,dx=0
\end{align}
due to the periodic boundary conditions.  
Thus, 
\begin{align}
\int_{\Omega}\theta\,dx =c_1, 
\end{align}
for some dimensional constant $c_1$.     
Integrating \eqref{Bouss_mo} over $\Omega$, we find 
\begin{align}
&\quad
\frac{d}{dt}\int_{\Omega}\bu\,dx
+ \int_{\Omega}\nabla\cdot(\bu\otimes\bu) \,dx
+\int_{\Omega}\nabla p\,dx
= \int_{\Omega}\nu\triangle\bu\,dx
+ \int_{\Omega}\bg\theta \,dx,
\end{align}
so that 
\begin{align}
\label{ud_int}
\frac{d}{dt}\int_{\Omega}\bu\,dx
= \int_{\Omega}\bg\theta \,dx,
\end{align}
Writing this in component form, this becomes
\begin{subequations}
\begin{align}
\label{int_zero_pre_i}
 \frac{d}{dt}\int_{\Omega}u_1\,dx &=0,
 \\\label{int_zero_pre_d}
 \frac{d}{dt}\int_{\Omega}u_2\,dx &= g\int_{\Omega}\theta\,dx=gc_1.
\end{align}
\end{subequations}
Thus, 
\begin{align}
\int_{\Omega}u_2\,dx = gc_1t+c_0.
\end{align}
Taking the $L^2$ inner product of \eqref{Bouss_mo}
with $\bu$, and of \eqref{Bouss_den}, using \eqref{symmetry}, and adding the
results, we obtain
\begin{align}
 \tfrac{1}{2}\tfrac{d}{dt}\|\bu\|_{L^2}^2+\nu\|\bu\|_{\h^1}^2 
 = \ip{\bg\theta}{\bu}
 \leq g\|\theta\|_{L^2}|\bu\|_{L^2}
\end{align}
Applying Poincar\'e's inequality and the Cauchy-Schwarz and Young's
inequalities, we find
\begin{align}
 \tfrac{1}{2}\tfrac{d}{dt}\|\bu\|_{L^2}^2+\nu\lambda_1\|\bu\|_{L^2}^2 
 \leq  \tfrac{g^2}{2\nu\lambda_1}\|\theta\|_{L^2}^2+ \tfrac{\nu\lambda_1}{2}\|\bu\|_{L^2}^2.
\end{align}
Thus,
\begin{align}
\tfrac{d}{dt}\|\bu\|_{L^2}^2+\nu\lambda_1\|\bu\|_{L^2}^2 
 \leq \tfrac{g^2}{\nu\lambda_1}\|\theta\|_{L^2}^2
 = \tfrac{1}{\nu\lambda_1}\|\thetaInit\|_{L^2}^2.
\end{align}
Integrating this equation, we have
\begin{align}
\|\bu(t)\|_{L^2}^2
&\leq 
e^{-\nu\lambda_1t}\|\buInit\|_{L^2}^2+\int_0^te^{-\nu\lambda_1(t-s)}\tfrac{g^2}{
2\nu\lambda_1}\|\thetaInit\|_{L^2}^2\,ds
 \\&= \notag
e^{-\nu\lambda_1t}\|\buInit\|_{L^2}^2+\frac{1-e^{-\nu\lambda_1t}}{\nu^2\lambda_1^2}
g^2\|\thetaInit\|_{L^2}^2.
\end{align}
In particular, $\bu\in L^\infty(0,T;H)$ for all $T>0$.  This implies
\begin{align}
 |gc_1t+c_0|
 &= \abs{ \int_{\Omega}u_2\,dx}
 \leq  \int_{\Omega}|u_2|\,dx
 \leq |\Omega|^{1/2}\|\bu\|_{L^2}
 \\&\notag
 \leq
|\Omega|^{1/2}\sqrt{e^{-\nu\lambda_1t}\|\buInit\|_{L^2}^2+\frac{1-e^{-\nu\lambda_1t}}{
2\nu^2\lambda_1^2}g^2\|\thetaInit\|_{L^2}^2}
\end{align}
Sending $t\maps\infty$, we see that we must have $c_1=0$.  Thus, 
\eqref{int_zero_pre_i} and \eqref{int_zero_pre_d} become
\begin{align}
 \frac{d}{dt}\int_{\Omega}\bu\,dx=\mathbf{0}, \qquad
\frac{d}{dt}\int_{\Omega}\theta\,dx=0.
\end{align}
Therefore, the flow preserves the spatial mean.

 \section*{Acknowledgement}
 \noindent
 This research was partially supported by NSF grant DMS14-25877 and the CNMS grant at UMBC (A. Biswas) and DMS11-09784 (C. Foias).  
 
 %~~~~~~~~~~~~~~~~~~~~~~~~~~~~~~~~~~~~~~~~~~~~~~~~~~~~~~~~~~~~~~~~~~~~
\begin{scriptsize}
\bibliographystyle{abbrv}%amsalpha%amsplain%plain
\bibliography{BFL}
\end{scriptsize}
%~~~~~~~~~~~~~~~~~~~~~~~~~~~~~~~~~~~~~~~~~~~~~~~~~~~~~~~~~~~~~~~~~~~~

%~~~~~~~~~~~~~~~~~~~~~~~~~~~~~~~~~~~~~~~~~~~~~~~~~~~~~~~~~~~~~~~~~~~~
\end{document}